\documentclass[11pt,a4paper]{article}
\usepackage{amssymb,amsmath,amsthm}
\usepackage{multirow}
\usepackage{blindtext}
\usepackage{amsfonts}
\usepackage{amsmath,amssymb,pxfonts, xcolor}
\usepackage{graphicx}
\usepackage{ulem}
\usepackage{cancel}
\usepackage{tabularx}

\usepackage{color}
\usepackage{soul}
\usepackage{float}
\usepackage{comment}
\usepackage{adjustbox}
\usepackage{booktabs}
\usepackage{arydshln}

\newtheorem{theorem}{Theorem}

\newtheorem{example}[theorem]{Example}

\newtheorem{proposition}[theorem]{Proposition}
\newtheorem{remark}[theorem]{Remark}

\numberwithin{equation}{section}

\makeatletter
\newcommand{\biggg}[1]{{\hbox{$\left#1\vbox to 20.5pt{}\right.\n@space$}}}
\newcommand{\Biggg}[1]{{\hbox{$\left#1\vbox to 23.5pt{}\right.\n@space$}}}
\newcommand{\bigggg}[1]{{\hbox{$\left#1\vbox to 26.5pt{}\right.\n@space$}}}
\newcommand{\Bigggg}[1]{{\hbox{$\left#1\vbox to 29.5pt{}\right.\n@space$}}}
\newcommand{\biggggg}[1]{{\hbox{$\left#1\vbox to 32.5pt{}\right.\n@space$}}}
\newcommand{\Biggggg}[1]{{\hbox{$\left#1\vbox to 35.5pt{}\right.\n@space$}}}
\newcommand{\bigggggg}[1]{{\hbox{$\left#1\vbox to 38.5pt{}\right.\n@space$}}}
\newcommand{\Bigggggg}[1]{{\hbox{$\left#1\vbox to 41.5pt{}\right.\n@space$}}}

\title{On the difference between the volatility swap strike and the zero vanna implied volatility}
\author{Elisa Al\`{o}s\thanks{Dpt. d'Economia i Empresa, Universitat Pompeu Fabra.} \quad Frido Rolloos\thanks{Ortec Finance.} \quad Kenichiro Shiraya\thanks{Graduate School of Economics, The University of Tokyo. Kenichiro Shiraya is supported by CARF.}}

\begin{document}
\maketitle

\begin{abstract}
In this paper, Malliavin calculus is applied to arrive at exact formulas for the difference between the volatility swap strike and the zero vanna implied volatility for volatilities driven by fractional noise.  To the best of our knowledge, our estimate is the first to derive the rigorous relationship between the zero vanna implied volatility and the volatility swap strike. 
In particular, we will see that the zero vanna implied volatility is a better approximation for the volatility swap strike than the ATMI.

\bigskip


\end{abstract}

\section{Introduction}
The pricing and hedging of volatility derivatives is an active and fruitful area of research in quantitative finance. One of the first volatility derivatives to be traded in the over-the-counter market is the variance swap. Another instrument to trade volatility is the volatility swap, which unlike the variance swap has a payoff that is linear in volatility.
However, volatility swaps are less liquid than variance swaps. The reason for this is because the price of a volatility swap was for a long time considered to be highly model-dependent. 

In their important paper, Carr and Lee (2008) challenged the idea that volatility swaps are highly model-dependent. In the case where the correlation between the volatility and the underlying asset is zero, they proved in their paper that the exact volatility swap strike is in fact model-free, and like the variance swap can be synthesised using a continuous strip of options. The difference is that for volatility swaps the replicating strip of options has to be continuously rebalanced. An elegant derivation of the replicating portfolio for volatility swaps is given by Friz and Gatheral (2005). When correlation deviates from zero,  there is indeed model dependence and only model-free approximations are possible. In this paper, when we speak of `model-free' we mean model-independent within the class of fractional stochastic volatility models.

In recent years, the fractional volatility models introduced by Comte and Renault (1998) have led to several papers which explore the valuation of volatility derivatives under the models. For example, Bergomi and Guyon (2011) and El Euch, Fukawasa, Gatheral and Rosenbaum (2019) derive approximation formulas for the variance swap strike by using expansion techniques.
Al\'{o}s and Shiraya (2019) approximate the volatility swap strike by immunising correlation dependence to first order and also provide an estimation method for the Hurst parameter from ATM implied volatilities.

While the aforementioned papers establish relationships between volatility derivatives and the ATM implied volatility, a different approach to model free approximate pricing of volatility swaps has been put forth by Rolloos and Arslan (2017). Using only the generalised Hull-White formula and Taylor expansions, they show that the volatility swap strike is approximately equal to the implied volatility at the strike where the Black-Scholes vanna of a vanilla option is zero. Like the Carr-Lee approximation, the Rolloos-Arslan approximation is to a large extent immune to correlation to first order. Furthermore, although the  two are not equal, numerical tests thus far have shown that both are of comparable accuracy.

In addition to being intuitive and straightforward to implement, another feature of the zero vanna implied volatility approximation is that it lends itself to rigorous quantification of the error between the true volatility swap price and the zero vanna implied volatility. 
This paper extends Rolloos and Arslan (2017) to general fractional volatility models and provides the rigorous relationship between the zero vanna implied volatility and the volatility swap strike for both the uncorrelated and correlated case.
We show that the zero vanna implied volatility is a better approximation for the volatility swap strike than both the ATM implied volatility and the approximation formula of Al\`{o}s and Shiraya (2019) for the cases we consider in this paper.

The paper is organised as follows. In Section 2 we introduce the relevant concepts and establish notation. Section 3 is devoted to deriving exact expressions for the difference between the volatility swap strike and the zero vanna implied volatility as well as the short time limit of the errors. Due to their length, all proofs of propositions and theorems in Section 3 have been placed in Appendix \ref{appendix2}. In Section 4 numerical examples are presented for a rough Bergomi model with various values of the Hurst parameter. Section 5 contains concluding remarks.

\section{The main problem and notations}

Consider a stochastic volatility model for the log-price of a stock
under a risk-neutral probability measure $P$:
\begin{equation}
\label{themodel}
X_{t}=X_0-\frac{1}{2}\int_{0}^{t}{\sigma _{s}^{2}}ds+\int_{0}^{t}\sigma
_{s}\left( \rho dW_{s}+\sqrt{1-\rho ^{2}}dB_{s}\right) ,\quad t\in \lbrack
0,T].  
\end{equation}%
Here, $X_0$ is the current log-price, $%
W$ and $B$ are independent standard Brownian motions defined on a complete probability
space  $(\Omega ,\mathcal{G},P)$, and $\sigma $ is a square-integrable and
right-continuous stochastic process adapted to the filtration generated by $%
W $.  We denote by $\mathcal{F}^{W}$ and $\mathcal{F}^{B}$
the filtrations generated by $W$ and $B$ and $\mathcal{F}:=%
\mathcal{F}^{W}\vee \mathcal{F}^{B}.$ We assume the interest rate $r$ to be zero for the sake of simplicity. The same arguments in this paper hold for $r\neq 0$.

Under the above model, the price of a European call with strike price $K$ is given by the equality
\[
V_{t}=E_{t}[(e^{X_{T}}-K)_{+}], 
\]%
where $E_{t}$ is the $\mathcal{F}_{t}-$conditional expectation with respect
to $P$ (i.e., $E_{t}[Z]=E[Z|\mathcal{F}_{t}]$). In the sequel, we make use
of the following notation:

\begin{itemize}
\item $v_t=\sqrt{\frac{Y_t}{T-t}}$, where $Y_t=\int_{t}^{T}\sigma _{u}^{2}du$.

That is, $v$
represents the future average volatility, and it is not an adapted process. 
Notice that $E_{t}\left[
v_{t}\right] $ is the fair strike of a volatility swap with maturity time $T.$

\item $BS(t,T,x,k,\sigma )$ is the price of a  European call option
under the classical Black-Scholes model with constant volatility $\sigma $,
stock price $e^x$, time to maturity $T-t,$ and strike $K=\exp
(k) $. Remember that (if $r=0$)

\[
BS(t,T,x,k,\sigma )=e^{x}N(d_1(k,\sigma ))-e^{k}N(d_2(k,\sigma )), 
\]%
where $N$ denotes the cumulative probability function of the standard normal
law and

\[
d_1\left( k,\sigma \right) :=\frac{x-k}{\sigma \sqrt{T-t}}+
\frac{\sigma }{2}\sqrt{T-t},  \hspace{0.4cm} d_2\left( k,\sigma \right) :=\frac{x-k}{\sigma \sqrt{T-t}}-
\frac{\sigma }{2}\sqrt{T-t}.
\]%
For the sake of simplicity we make use of the notation $BS(k,\sigma):=BS(t,T,X_t,k,\sigma )$.
\item  The inverse function $BS^{-1}(t,T,x,k, \cdot)$ of the Black-Scholes formula with respect to the volatility parameter is defined as
\[
BS(t,T,x,k, BS^{-1}( t,T,x,k,\lambda) )=\lambda,
\]
for all $\lambda>0$.
For the sake of simplicity, we denote $BS^{-1}(k,\lambda)\ :=BS^{-1}(t,T,X_{t},k,\lambda)$.

\item For any fixed $t,T,X_{t},k,$ we define the implied volatility $I(t,T,X_{t},k) $ as the quantity such that 
\[
BS( t,T,X_{t},k,I( t,T,X_{t},k) ) =V_{t}.
\]
Notice that $I(t,T,X_t,k)=BS^{-1}( t,T,X_t,k,V_t)$.

\item $\hat{k_t}$ is the {\it zero vanna strik}e at time $t$. That is, the strike such that 
$$d_2(\hat{k}_t,I(t,T,X_t,\hat{k}_t))=0.$$ Moreover, we will refer to $I(t,T,X_t,\hat{k}_t)$ as the {\it zero vanna implied volatility}. We recall that (Black-Scholes) vanna is the partial derivative of the Black-Scholes delta with respect to implied volatility, which is directly proportional to $d_2$. The zero vanna strike is therefore the strike where the Black-Scholes vanna of a vanilla option is zero.

\item $
\Lambda _{r}:=E_{r}\left[ BS\left( t,T,X_t,\hat{k}_t,v_{t}\right) \right]$.

\item $
\Theta_{r}(k):=BS^{-1}(k,\Lambda_r)$. Notice that $\Theta_{t}(\hat{k}_t)=I(t,X_t,\hat{k}_t,V_t)$ and $\Theta_{T}(\hat{k}_t)=v_t$.

\item $G(t,T,x,k,\sigma ):=( \frac{\partial ^{2}}{\partial x^{2}}-\frac{%
\partial}{\partial x}) BS(t,T,x,k,\sigma )$.
\item $H(t,T,x,k,\sigma ):=( \frac{\partial ^{3}}{\partial x^{3}}-\frac{%
\partial ^{2}}{\partial x^{2}}) BS(t,T,x,k,\sigma )$.

\end{itemize}

In the remaining of this paper $\mathbb{D}_{W}^{1,2}$ denotes the domain of the
Malliavin derivative operator $D^{W}$ (see Appendix \ref{appendix1}) with respect to the Brownian motion $%
W. $ We also consider the
iterated derivatives $D^{n,W}$ , for $n>1,$ whose domains will be denoted by 
$\mathbb{D}_{W}^{n,2}$. We will use the notation $\mathbb{L}_{W}^{n,2}=$\ $%
L^{2}(\left[ 0,T\right] ;\mathbb{D}_{W}^{n,2})$.

\section{Limit theorems for the zero vanna implied volatility}\label{sec3}

The proofs of Propositions and Theorems of this section are in Appendix \ref{appendix2}.

\subsection{The uncorrelated case}

Let us consider the following hypotheses:

\begin{description}
\item[(H1)] There exist two positive constants $a,b$ such that $a\leq \sigma
_{t}\leq b,$ for all $t\in \left[ 0,T\right] .$

\item[(H2)] $\sigma^2\in \mathbb{L}^{2,2}_W$ and there exist two constants $\nu>0$ and $H\in\left(0,1\right)$ such that, for all $0<r, \theta<s<T$,
$$
|E_r[D_r^W\sigma_s^2]|\leq \nu(s-r)^{H-\frac12}, \hspace{0.3cm} |E_r[D_\theta^W D_r^W\sigma_s^2]|\leq \nu^2(s-r)^{H-\frac12}(s-\theta)^{H-\frac12}.
$$

\end{description}
The key tool in our analysis will be the following relationship between the zero vanna implied volatility and the fair strike of a volatility swap.

\begin{proposition}
\label{General}
Consider the model (\ref{themodel}) with $\rho =0$ and assume
that hypotheses (H1) and (H2) hold. Then the zero vanna implied volatility
admits the representation%
\begin{eqnarray}
\label{implied}
I\left( t,T,X_t,\hat{k}_t\right) &=& E_{t}\left[ v_{t}\right] \nonumber\\
&&+ E_t\Bigg[\int_{t}^{T}\left( BS^{-1}\left( \hat{k}_t, \Lambda_r\right)\right) ^{\prime \prime\prime} (D^-A)_r U_r dr\Bigg]\nonumber\\
&&+ \frac12 E_t\Bigg[
\int_{t}^{T}\left( BS^{-1}\left( \hat{k}_t, \Lambda_r\right)\right) ^{(iv)} A_r U_r^2 dr \Bigg],
\label{I-E[v]}
\end{eqnarray}
where 
$$A_r:=\frac12 \int_{r}^{T} U_{s}^{2}ds, \hspace{0.3cm}  (D^-A)_r:=\frac{1}{2}\int^T_r D_r^W U^2_s ds,$$  and
\begin{eqnarray}
U_{r}&:=&E_{r}\left[D_r^W\left(BS(t,T,X_t,\hat{k}_t,v_t)\right)\right]\nonumber\\
&=&E_r\left[\frac{\partial BS}{\partial \sigma}(t,T,X_t,\hat{k}_t,v_t)\frac{1}{2v_t(T-t)}\int_r^T D_r^W\sigma_s^2 ds\right].
\end{eqnarray}
\end{proposition}

\vspace{0.2cm}

In order to prove our limit results, we will need the following hypothesis.

\begin{description}
\item[(H2')]  $\sigma\in \mathbb{L}^{3,2}_W$ and there exists two constants $\nu>0$ and $H\in (0,1)$ such that, for all $0<r<u,s,\theta<T$
$$
|E_r[D_r^W\sigma_s^2]|\leq \nu(s-r)^{H-\frac12}, \hspace{0.3cm} |E_r[D_\theta^WD_r^W\sigma_s^2]|\leq \nu^2(s-r)^{H-\frac12}(s-\theta)^{H-\frac12},
$$
and 
$$
|E_r[D_u^W D_\theta^W D_r^W\sigma_s^2]|\leq \nu^3(s-r)^{H-\frac12}(s-\theta)^{H-\frac12}(s-u)^{H-\frac12}.
$$
\end{description}

\begin{theorem}
\label{uncorrelated}
Consider the model (\ref{themodel}) and assume that hypotheses (H1) and (H2') hold. Then, 
\begin{eqnarray*}
I\left( t,T,X_t,\hat{k}_t\right) -E_{t}\left[v_{t}\right] 
= O(\nu^4(T-t)^{4H+1 }).
\end{eqnarray*}

\end{theorem}

\subsection{The correlated case}

We will consider the following hypothesis.

\begin{description}
\item[(H3)] Hypotheses (H1), (H2'), hold and the terms 
\[
\frac{1}{(T-t)^{3+2H}}E_{t} \left[\left(\int_{t}^{T}
\int_{s}^{T}D_{s}^{W}\sigma _{r}^{2}dr ds\right) ^{2}\right] ,
\]%
\[
\frac{1}{(T-t)^{2+2H }}E_{t}\left[
\int_{t}^{T} \int_{s}^{T}D_{s}^{W}\sigma_{r} \int_{r}^{T}D_{r}^{W}\sigma _{u}^{2}du dr ds
\right],
\]%
\[
\frac{1}{(T-t)^{2+2H }}E_{t}\left[
\int_{t}^{T} \int_{s}^{T} \int_{r}^{T}D_{s}^{W}D_{r}^{W}\sigma_{u}^{2}du dr ds\right],
\]%
have a finite limit as $T\rightarrow t.$

\end{description}

The following result, that follows from the same arguments as Proposition 4.1  in Al\`os and Shiraya (2019), gives us an exact decomposition for the zero vanna implied volatility that will be the main tool in this section.
\begin{proposition}
\label{Theoremcorrelatedcase}Consider the model \eqref{themodel} and assume
that hypotheses (H1), (H2') and hold for some $H\in (0,1)$. Then, 
\begin{eqnarray}
I\left( t,T,X_t,\hat{k}_t\right) &=&I^{0}\left( t,T,X_t,\hat{k}_t,\right)
\nonumber \\
&&+\frac{\rho }{2}E_t \left[\int_{t}^{T}( BS^{-1}) ^{\prime }( \hat{k}_t,\Gamma _{s}) H(s,T,X_{s},\hat{k}_t,v_{s})\zeta _{s}ds\right],
\label{impliedrel}
\end{eqnarray}%
where $I^{0}( t,T,X_t,\hat{k}_t) $ denotes the zero vanna implied
volatility in the uncorrelated case $\rho =0$,
\[
\Gamma _{s}:=E_{t}[ BS(t,T,X_t,\hat{k}_t,v_{t})] +\frac{\rho }{2}%
E_{t}\left[\int_{t}^{s} H(r,T,X_{r},\hat{k}_t,v_{r})\zeta _{r}dr\right], 
\]%
and $\zeta _{t}:=\sigma _{t}\int_{t}^{T}D_{t}^{W}\sigma _{r}^{2}dr.$ 
\end{proposition}

Theorem \ref{uncorrelated} and Proposition \ref{Theoremcorrelatedcase} allow us to prove the following result.
\begin{theorem}
\label{themaintheorem}
Consider the model \eqref{themodel}\ and assume
that hypotheses (H1), (H2') and (H3)  hold for some $H\in(0,1) $. Then%

\begin{eqnarray}
\label{lim1}
\lefteqn{\lim_{T\rightarrow t}\frac{I( t,T,X_t,\hat{k}_t) -E_{t}[v_{t}]} {(T-t)^{2H }}}\nonumber \\
&=& \lim_{T\rightarrow t}\frac{3\rho^{2}}{8\sigma _{t}^{3}(T-t)^{3+2H}}
E_{t}\left[\left(\int_{t}^{T}\int_{s}^{T}D_{s}^{W}\sigma _{r}^{2}dr ds\right)^{2} \right]\nonumber\\
&&-\lim_{T\rightarrow t}\frac{\rho ^{2}}{{2\sigma _{t}^2(T-t)^{2+2H}}}E_{t}\left[ 
\int_{t}^{T} \int_{s}^{T}D_{s}^{W}\sigma_{r} \int_{r}^{T}D_{r}^{W}\sigma _{u}^{2}du dr ds\right]\nonumber\\
&&-\lim_{T\rightarrow t}\frac{\rho ^{2}}{2\sigma _{t}(T-t)^{2+2H}}
E_{t}\left[\int_{t}^{T} \int_{s}^{T}\int_{r}^{T}D_{s}^{W}D_{r}^{W}\sigma _{u}^{2}du dr ds\right].
\end{eqnarray}
\end{theorem}

\begin{example}
Let $W^H$ be a Riemann-Liouville fractional Brownian motion (RLfBm) with Hurst parameter $H\in (0,1)$  defined in a time interval $[0,T]$. That is, 
$$
W_t^H:=\frac{1}{\Gamma(H+\frac12)}\int^t_0 (t-s)^{H-\frac12} dW_s.
$$
Assume that $\sigma_t=f(W^H_t)$, where $f\in\mathcal{C}_b^3$ with a range in a compact set of $\mathbb{R}^+$  and $W^H_t$ is a fBm with Hurst parameter $H$. Then 
the above result proves that, in the correlated case
\begin{equation}
\label{order}
I(t,T,X_t,\hat{k}_t) -E_{t}[ v_{t}]=O((T-t)^{2H}).
\end{equation}
\end{example}


\begin{remark}
Notice that the term $T_2^{1,2}$ in the proof of Theorem \ref{themaintheorem} in Appendix \ref{appendix2} is of the order $(\rho(T-t)^{\frac12+{2}H})$. When $T-t$ does not tend to zero, this term can not be neglected.
\end{remark}

\begin{remark}
Hypotheses (H1)-(H3) have been chosen for the sake of simplicity. The same results can be extended to other stochastic volatility models (see e.g., Section 5 in Al\`{o}s and Shiraya (2019)).
\end{remark}

\begin{remark}

In the case $H=\frac12$, Theorem 4.2  in Al\`os and Shiraya (2019) gives us that

\begin{eqnarray*}
\lefteqn{\lim_{T\rightarrow t}\frac{I( t,T,X_t,k^*_t) -E_{t}[v_{t}]}{(T-t)}}\nonumber \\
&=& \lim_{T\rightarrow t}\frac{3\rho ^{2}}{8\sigma _{t}^{3}(T-t)^{4 }}E_{t}\left[ \left(\int_{t}^{T}\int_{s}^{T}D_{s}^{W}\sigma _{r}^{2}dr ds\right)^{2}\right] \nonumber \\
&&
-\lim_{T\rightarrow t}\frac{\rho ^{2}}{{2\sigma _{t}^2 (T-t)^3}}E_{t}\left[ 
\int_{t}^{T} \int_{s}^{T}D_{s}^{W}\sigma_{r} \int_{r}^{T}D_{r}^{W}\sigma _{u}^{2}du dr ds
\right] \nonumber\\
&&-\lim_{T\rightarrow t}\frac{\rho ^{2}}{2\sigma _{t}(T-t)^{3 }}E_{t} \left[\int_{t}^{T} \int_{s}^{T}\int_{r}^{T}D_{s}^{W}D_{r}^{W}\sigma _{u}^{2}du dr ds\right]\nonumber\\
&&+\lim_{T\rightarrow t}\frac{\rho}{4(T-t)^{2}}E_t\left[\int_{t}^{T} \int_{s}^{T}D_{s}^{W}\sigma_{r}^{2}dr ds\right],  a.s.
\end{eqnarray*}

where $k^*_t$ denotes the at-the-money strike. This, jointly with Theorem \ref{themaintheorem} implies that, if $H=\frac12$
\begin{eqnarray*}
\lefteqn{\lim_{T\rightarrow t}\frac{I( t,T,X_t,k^*_t) -I( t,T,X_t,\hat{k}_t)}{(T-t)}}\nonumber\\
&&=\lim_{T\rightarrow t}\frac{\rho}{4(T-t)^{2}}E_t\left[\int_{t}^{T} \int_{s}^{T}D_{s}^{W}\sigma_{r}^{2}dr ds\right], a.s.
\end{eqnarray*}
Now, the term in the right-hand side can be written in terms of the short-time limit skew  (see Al\`os, Le\'{o}n and Vives \cite{ALV}). This allows us to write
\begin{eqnarray*}
\label{greater0}
\lim_{T\rightarrow t}\frac{I( t,T,X_t,k^*_t) -I( t,T,X_t,\hat{k}_t)}{(T-t)}=\frac{\sigma _{t}^{2}}{2}%
\lim_{T\rightarrow t}\frac{\partial I}{\partial k}(
t,T,X_t,k^*_t) 
\end{eqnarray*}%
and then we obtain the following model free relationship between the zero vanna implied volatility, the at-the-money implied volatility and the implied volatility skew:
\begin{equation}
\label{rel}
I( t,T,X_t,k^*_t) -I( t,T,X_t,\hat{k}_t)\approx \frac{I^2( t,T,X_t,k^*_t)(T-t)}{2}%
\frac{\partial I}{\partial k}(
t,T,X_t,k^*_t).
\end{equation}

\end{remark}

\section{Numerical examples}\label{section4}

In this section we confirm the validity of our estimates by using numerical examples, 
and discuss the potential applications of our results.
Let us consider a rough Bergomi model as in Bayer, Friz and Gatheral (2016). That is, assume that the volatility process is given by
\begin{eqnarray}
S_t &=& \exp\left( X_0- \frac{1}{2} \int_{0}^{t} \sigma _{s}^{2} ds+\int_{0}^{t}\sigma
_{s}\left( \rho dW_{s}+\sqrt{1-\rho^{2}} dB_{s}\right)\right),\\
\sigma_t^2 &=&\sigma_0^2\exp\left(\alpha W_t^H-\frac12 \alpha^2t^{2H}\right), t\in[0,T],
\end{eqnarray}
for some positive real values $\sigma_0^2$ and $\alpha$, $\rho \in [-1,1]$, $W_t^H:=\sqrt{2H}\int_0^t (t-s)^{H-\frac12}dW_s$, and $H\in (0,1)$. $W$ and $B$ are independent standard Brownian motions.

A direct computation gives us (see Bayer, Friz and Gatheral (2016)) that, for all $s<t$
\begin{eqnarray}
\label{rBcov}
E(W_t^H W_s^H)
&=& 
s^{2H} \int^1_0 \frac{2H}{(1-x)^{\frac12 - H}(t/s-x)^{\frac12 - H}}dx.
\end{eqnarray}
Moreover, for all $s, t \ge 0$
\begin{equation}
\label{rBW}
E(W_t^H B_s)=\frac{\rho\sqrt{2H}}{H+\frac12}\left(t^{H+\frac12}-(t-\min(t,s))^{H+\frac12}\right).
\end{equation}

Notice that even when this model does not satisfy (H1), (H2) and (H2'), the limit results in Theorem \ref{themaintheorem} are still valid. To see this we make use of an approximation argument. Let us define $\phi(x):=\sigma_0\exp(x)$. For every $n>1$, consider a function $\phi_n\in \mathcal{C}_b^2$ such that $\phi_{n} (x)=\phi(x)$ for any $x\in [-n,n]$, $\phi_{n} (x)\in[\phi(-2n)\vee \phi(x), \phi(-n)]$ for $x\le -n$, and $\phi_n (x)\in [\phi(n),\phi(x)\wedge \phi(2n)]$ for $x\ge n$. Then
$$
(\sigma_s^n)^2:=\phi_n\left( \alpha W_{t}^H-\frac12 \alpha^2t^{2H} \right).
$$
It is easy to see that $\sigma_s^n $ satisfies (H1), (H2) and (H2'). Then, we can write (we consider $t=0$ for the sake of simplicity)
\begin{eqnarray*}
\lefteqn{I( 0,T,X_0,\hat{k}_0) -E[v_{0}]}\nonumber \\
&=& I( 0,T,X_0,\hat{k}_0) -I^n( 0,T,X_0,\hat{k}_0)\nonumber \\
&&+ I^n(0,T,X_0,\hat{k}_0) -E[v^n_{0}]\nonumber \\
&&+E\left[v^n_{0}\right]-E[v_{0}]\nonumber \\
&=:& T_1+T_2+T_3,
\end{eqnarray*}
where $I^n$ and $E_{t}[v^n_{t}]$ denote the implied volatility and the fair price of the volatility swap under the volatility process $\sigma^n$, respectively.  Now, the proof of Theorem \ref{themaintheorem} gives us that  (in the correlated case $\rho\neq 0$)
$T_2 = O(T^{2H})$. On the other hand, following similar arguments as in Al\`os and Shiraya (2019), we can see that for $n$ large enough $T_1$ and $T_3$ in the proof of Theorem \ref{themaintheorem} are of a higher order than $T_2$. 

For the numerical simulations we set $S_0 = e^{X_0}=100$, 
$\sigma _{0} = 20\%$, $\alpha =0.8$, correlation values $\rho =0$ and $-0.8$, Hurst parameter values $H=0.1$, $0.3$, $0.5$, $0.6$, $0.9$, 
and times to maturity $0.25$, $0.5$, $1$, $2$, and $3$ years as base settings. In order to calculate the implied volatilities and volatility swap prices, we use Monte Carlo simulation with 500 time steps for one year and twenty million trials.

The discretized fractional Brownian motion $W^H$ and standard Brownian motion $B$ are made by the standard normal random variables $Z$ and the triangular matrix obtained by the Cholesky decomposition of the covariance matrix.
In more detail, 
firstly, we made independent standard normal random variables $Z_{1},\cdots,Z_{2m}$ ($m$ is the number of time steps) and the $2m\times 2m$ covariance matrix $\Sigma$ whose $(i,j)$ element $\Sigma_{i,j}$ for $0<t_1=t_{m+1}<t_2=t_{m+2}<\cdots<t_m=t_{2m}=T$ is expressed as follows:
\begin{eqnarray}
\Sigma_{i,j} = 
\begin{cases}
\min(t_i,t_j)^{2H} \int^1_0 \frac{2H dx}{(1-x)^{\frac12 - H}(\max(t_i,t_j)/\min(t_i,t_j)-x)^{\frac12 - H}}
& \text{if}\ 0<i,j \le m, \\
\rho \frac{\sqrt{2H}}{H+1/2} \left(t_i^{H+\frac12} - (t_i - \min(t_i,t_j))^{H+\frac12} \right)
& \text{if}\ 0<i\le m<j\le 2m, \\
\rho \frac{\sqrt{2H}}{H+1/2} \left(t_j^{H+\frac12} - (t_j - \min(t_i,t_j))^{H+\frac12} \right)
& \text{if}\ 0<j\le m<i\le 2m, \\
\min(t_i,t_j) & \text{otherwise.}
\end{cases}
\end{eqnarray}
Next, we make $2m \times 2m$ triangular matrices $C$ and $C^T$ obtained by the Cholesky decomposition of $\Sigma$ as $\Sigma = CC^T$.
Then, the discretized fractional Brownian motion $W^H_{t_1},\cdots,W^H_{t_m}$ and the standard Brownian motion $B_{t_1},\cdots, B_{t_m}$ are calculated as 
$(W^H_{t_1},\cdots,W^H_{t_m}, B_{t_1},\cdots, B_{t_m})'=C (Z_{1},\cdots,Z_{2m})'$.

We apply the Euler-Maruyama scheme for the exponent of both the underlying asset price and its volatility processes, and calculate the European put option premium
\begin{eqnarray}
E\left[\left(K - \exp \left(X_0 
- \frac{1}{2}\sum^{m-1}_{i=0} \sigma^2_{t_i}(t_{i+1} - t_{i}) 
+ \sum^{m-1}_{i=0} \sigma_{t_i}\left(\rho (W_{t_{i+1}} - W_{t_{i}})  + \sqrt{1-\rho^2} (B_{t_{i+1}} - B_{t_{i}}) \right)\right) \right)_+\right],
\end{eqnarray}
and the volatility swap strike
\begin{eqnarray}
E\left[\sqrt{\frac{1}{T}\sum^{m-1}_{i=0} \sigma_{t_i}^2(t_{i+1} - t_{i})}\,\right],
\end{eqnarray}
where $K$ is the strike price of the European put option.

To increase accuracy, the Black-Scholes model has been used as the control variate for the Monte Carlo simulations to obtain the option premiums.
Once the exact volatility swap strikes and options prices have been calculated, the bisection method is used to infer implied volatilities, including zero vanna implied volatilities.
To compare our new results to the approximation formula (4.8) of Al\`{o}s-Shiraya (2019), 
we also calculate the ATM skew ($\frac{\partial I}{\partial k}$) using the difference method on the implied volatilities.

Tables \ref{table:t1} and \ref{table:t2} below show the results of the uncorrelated case and correlated case, respectively.
Also, we examine the stressed parameter cases.
Tables \ref{table:t3} and \ref{table:t4} show the results of the stressed $\sigma_0$ which is twice the size of the base setting (i.e. $\sigma_0=40\%$).
Tables \ref{table:t5} and \ref{table:t6} show the results of the stressed $\alpha$ which is 2.5 times the size of the base setting (i.e. $\alpha=2$).

In the tables, ``VS'' is the simulated volatility swap value,
``IV ($\hat{k}$)'' and ``ATMI'' are the implied volatility at respectively the zero vanna strike 
and ATM strike, and
``AS(4.8)'' is the value of the formula (4.8) in Al\`{o}s and Shiraya (2019).
Note that in the uncorrelated case AS(4.8) and ATMI are equal because the ATM skew in the uncorrelated case is $0$.

We also show the difference between ``VS" and ``IV ($\hat{k}$)'', ``ATMI'', ``AS(4.8)'' respectively.
The standard deviations of the Monte Carlo simulations have been relegated to tables in Appendix \ref{appendix3}.

\begin{table}[H]
\caption{Zero-vanna approximation for volatility swaps ($\rho=0$, $\sigma_0=20\%$, $\alpha=0.8$)}
\label{table:t1}
\newcolumntype{Y}{>{\centering\arraybackslash}X}
\newcolumntype{Z}{>{\raggedleft\arraybackslash}X}
\begin{tabularx}{\linewidth}{XXYYYYY} \hline
$H$  &  $T$    & VS   & IV($\hat{k}$)      & ATMI    & VS - IV($\hat{k}$) & VS - ATMI  \\ \hline\hline
0.1 & 0.25 & 19.70\% & 19.69\% & 19.69\% & 0.00\%     & 0.00\%        \\
    & 0.5  & 19.65\% & 19.65\% & 19.65\% & 0.00\%     & 0.00\%        \\
    & 1.0  & 19.60\% & 19.60\% & 19.59\% & 0.00\%     & 0.00\%        \\
    & 2.0  & 19.54\% & 19.54\% & 19.53\% & 0.00\%     & 0.01\%        \\
~   & 3.0  & 19.50\% & 19.50\% & 19.49\% & 0.00\%     & 0.01\%        \\  \hdashline
0.3 & 0.25 & 19.75\% & 19.75\% & 19.75\% & 0.00\%     & 0.00\%        \\
    & 0.5  & 19.63\% & 19.63\% & 19.63\% & 0.00\%     & 0.00\%        \\
    & 1.0  & 19.44\% & 19.44\% & 19.43\% & 0.00\%     & 0.00\%        \\
    & 2.0  & 19.15\% & 19.15\% & 19.13\% & 0.00\%     & 0.02\%        \\
~   & 3.0  & 18.93\% & 18.92\% & 18.89\% & 0.01\%     & 0.03\%        \\ \hdashline
0.5 & 0.25 & 19.87\% & 19.87\% & 19.87\% & 0.00\%     & 0.00\%        \\
    & 0.5  & 19.74\% & 19.74\% & 19.74\% & 0.00\%     & 0.00\%        \\
    & 1.0  & 19.48\% & 19.48\% & 19.48\% & 0.00\%     & 0.00\%        \\
    & 2.0  & 18.98\% & 18.98\% & 18.96\% & 0.01\%     & 0.02\%        \\
~   & 3.0  & 18.51\% & 18.50\% & 18.46\% & 0.02\%     & 0.05\%        \\ \hdashline
0.7 & 0.25 & 19.94\% & 19.93\% & 19.93\% & 0.00\%     & 0.00\%        \\
    & 0.5  & 19.83\% & 19.83\% & 19.83\% & 0.00\%     & 0.00\%        \\
    & 1.0  & 19.56\% & 19.56\% & 19.55\% & 0.00\%     & 0.00\%        \\
    & 2.0  & 18.87\% & 18.86\% & 18.84\% & 0.01\%     & 0.03\%        \\
~   & 3.0  & 18.10\% & 18.07\% & 18.03\% & 0.03\%     & 0.07\%        \\ \hdashline
0.9 & 0.25 & 19.97\% & 19.97\% & 19.97\% & 0.00\%     & 0.00\%        \\
    & 0.5  & 19.89\% & 19.89\% & 19.89\% & 0.00\%     & 0.00\%        \\
    & 1.0  & 19.62\% & 19.63\% & 19.62\% & 0.00\%     & 0.00\%        \\
    & 2.0  & 18.76\% & 18.75\% & 18.73\% & 0.01\%     & 0.03\%        \\
~   & 3.0  & 17.64\% & 17.58\% & 17.54\% & 0.06\%     & 0.09\%       \\ \hline
\end{tabularx}
\end{table}

\begin{table}[H]
\caption{Zero-vanna approximation for volatility swaps ($\rho=-0.8$, $\sigma_0=20\%$, $\alpha=0.8$)}
\label{table:t2}
\newcolumntype{Y}{>{\centering\arraybackslash}X}
\newcolumntype{Z}{>{\raggedleft\arraybackslash}X}
\begin{tabularx}{\linewidth}{XXccccccc} \hline
$H$   & $T$    & VS   &  IV($\hat{k}$)     & ATMI    & AS(4.8) & VS - IV($\hat{k}$) & VS - ATMI & VS - AS(4.8)  \\ \hline\hline
0.1 & 0.25 & 19.70\% & 19.53\% & 19.41\% & 19.53\% & 0.16\%     & 0.28\%       & 0.16\%           \\
    & 0.5  & 19.65\% & 19.46\% & 19.29\% & 19.46\% & 0.18\%     & 0.36\%       & 0.19\%           \\
    & 1.0  & 19.60\% & 19.39\% & 19.12\% & 19.38\% & 0.21\%     & 0.48\%       & 0.22\%           \\
    & 2.0  & 19.54\% & 19.29\% & 18.89\% & 19.28\% & 0.25\%     & 0.64\%       & 0.26\%           \\
~   & 3.0  & 19.50\% & 19.23\% & 18.73\% & 19.21\% & 0.27\%     & 0.76\%       & 0.29\%           \\\hdashline
0.3 & 0.25 & 19.75\% & 19.63\% & 19.52\% & 19.62\% & 0.13\%     & 0.23\%       & 0.13\%           \\
    & 0.5  & 19.63\% & 19.43\% & 19.25\% & 19.43\% & 0.19\%     & 0.37\%       & 0.19\%           \\
    & 1.0  & 19.44\% & 19.15\% & 18.84\% & 19.14\% & 0.29\%     & 0.59\%       & 0.30\%           \\
    & 2.0  & 19.15\% & 18.71\% & 18.22\% & 18.68\% & 0.44\%     & 0.93\%       & 0.47\%           \\
~   & 3.0  & 18.93\% & 18.37\% & 17.73\% & 18.32\% & 0.56\%     & 1.20\%       & 0.60\%           \\\hdashline
0.5 & 0.25 & 19.87\% & 19.80\% & 19.73\% & 19.80\% & 0.06\%     & 0.14\%       & 0.06\%           \\
    & 0.5  & 19.74\% & 19.61\% & 19.46\% & 19.61\% & 0.13\%     & 0.28\%       & 0.13\%           \\
    & 1.0  & 19.48\% & 19.23\% & 18.94\% & 19.22\% & 0.25\%     & 0.54\%       & 0.26\%           \\
    & 2.0  & 18.98\% & 18.48\% & 17.97\% & 18.45\% & 0.51\%     & 1.01\%       & 0.54\%           \\
~   & 3.0  & 18.51\% & 17.77\% & 17.10\% & 17.72\% & 0.75\%     & 1.41\%       & 0.80\%           \\\hdashline
0.7 & 0.25 & 19.94\% & 19.91\% & 19.85\% & 19.90\% & 0.03\%     & 0.08\%       & 0.03\%           \\
    & 0.5  & 19.83\% & 19.75\% & 19.63\% & 19.75\% & 0.08\%     & 0.20\%       & 0.08\%           \\
    & 1.0  & 19.56\% & 19.35\% & 19.09\% & 19.34\% & 0.21\%     & 0.46\%       & 0.21\%           \\
    & 2.0  & 18.87\% & 18.33\% & 17.83\% & 18.30\% & 0.54\%     & 1.04\%       & 0.57\%           \\
~   & 3.0  & 18.10\% & 17.19\% & 16.53\% & 17.14\% & 0.91\%     & 1.57\%       & 0.96\%           \\\hdashline
0.9 & 0.25 & 19.97\% & 19.95\% & 19.92\% & 19.95\% & 0.01\%     & 0.05\%       & 0.01\%           \\
    & 0.5  & 19.89\% & 19.84\% & 19.75\% & 19.84\% & 0.05\%     & 0.14\%       & 0.05\%           \\
    & 1.0  & 19.62\% & 19.46\% & 19.22\% & 19.45\% & 0.17\%     & 0.40\%       & 0.17\%           \\
    & 2.0  & 18.76\% & 18.19\% & 17.71\% & 18.16\% & 0.57\%     & 1.05\%       & 0.60\%           \\
~   & 3.0  & 17.64\% & 16.58\% & 15.97\% & 16.53\% & 1.06\%     & 1.67\%       & 1.11\%      \\ \hline    
\end{tabularx}
\end{table}

\begin{table}[H]
\caption{Zero-vanna approximation for volatility swaps ($\rho=0$, $\sigma_0=40\%$, $\alpha=0.8$)}
\label{table:t3}
\newcolumntype{Y}{>{\centering\arraybackslash}X}
\newcolumntype{Z}{>{\raggedleft\arraybackslash}X}
\begin{tabularx}{\linewidth}{XXYYYYY} \hline
$H$  &  $T$    & VS   & IV($\hat{k}$)      & ATMI    & VS-IV($\hat{k}$) & VS-ATMI  \\ \hline\hline
0.1 & 0.25 & 39.39\% & 39.39\% & 39.38\% & 0.00\%     & 0.01\%        \\
    & 0.5  & 39.30\% & 39.30\% & 39.29\% & 0.00\%     & 0.01\%        \\
    & 1.0  & 39.20\% & 39.19\% & 39.16\% & 0.00\%     & 0.03\%        \\
    & 2.0  & 39.08\% & 39.07\% & 39.00\% & 0.01\%     & 0.08\%        \\
~   & 3.0  & 39.00\% & 38.99\% & 38.88\% & 0.01\%     & 0.12\%        \\\hdashline
0.3 & 0.25 & 39.51\% & 39.50\% & 39.50\% & 0.00\%     & 0.01\%        \\
    & 0.5  & 39.25\% & 39.25\% & 39.24\% & 0.00\%     & 0.01\%        \\
    & 1.0  & 38.87\% & 38.87\% & 38.83\% & 0.00\%     & 0.04\%        \\
    & 2.0  & 38.30\% & 38.27\% & 38.16\% & 0.03\%     & 0.14\%        \\
~   & 3.0  & 37.85\% & 37.80\% & 37.59\% & 0.05\%     & 0.26\%        \\\hdashline
0.5 & 0.25 & 39.74\% & 39.74\% & 39.73\% & 0.00\%     & 0.01\%        \\
    & 0.5  & 39.48\% & 39.48\% & 39.47\% & 0.00\%     & 0.01\%        \\
    & 1.0  & 38.96\% & 38.96\% & 38.92\% & 0.00\%     & 0.04\%        \\
    & 2.0  & 37.97\% & 37.92\% & 37.79\% & 0.05\%     & 0.17\%        \\
~   & 3.0  & 37.02\% & 36.90\% & 36.65\% & 0.12\%     & 0.37\%        \\\hdashline
0.7 & 0.25 & 39.87\% & 39.87\% & 39.87\% & 0.00\%     & 0.00\%        \\
    & 0.5  & 39.66\% & 39.66\% & 39.65\% & 0.00\%     & 0.01\%        \\
    & 1.0  & 39.11\% & 39.11\% & 39.08\% & 0.00\%     & 0.03\%        \\
    & 2.0  & 37.74\% & 37.67\% & 37.54\% & 0.07\%     & 0.20\%        \\
~   & 3.0  & 36.20\% & 35.97\% & 35.71\% & 0.23\%     & 0.49\%        \\\hdashline
0.9 & 0.25 & 39.94\% & 39.94\% & 39.94\% & 0.00\%     & 0.00\%        \\
    & 0.5  & 39.78\% & 39.78\% & 39.78\% & 0.00\%     & 0.00\%        \\
    & 1.0  & 39.25\% & 39.25\% & 39.22\% & 0.00\%     & 0.03\%        \\
    & 2.0  & 37.52\% & 37.43\% & 37.30\% & 0.09\%     & 0.22\%        \\
~   & 3.0  & 35.27\% & 34.90\% & 34.64\% & 0.37\%     & 0.63\%       \\\hline
\end{tabularx}
\end{table}

\begin{table}[H]
\caption{Zero-vanna approximation for volatility swaps ($\rho=-0.8$, $\sigma_0=40\%$, $\alpha=0.8$)}
\label{table:t4}
\newcolumntype{Y}{>{\centering\arraybackslash}X}
\newcolumntype{Z}{>{\raggedleft\arraybackslash}X}
\begin{tabularx}{\linewidth}{XXccccccc} \hline
$H$   & $T$    & VS   &  IV($\hat{k}$)     & ATMI    & AS(4.8) & VS -  IV($\hat{k}$) & VS - ATMI & VS - AS(4.8)  \\ \hline\hline
0.1 & 0.25 & 39.39\% & 39.06\% & 38.59\% & 39.05\% & 0.33\%     & 0.80\%       & 0.34\%           \\
    & 0.5  & 39.30\% & 38.92\% & 38.21\% & 38.90\% & 0.38\%     & 1.09\%       & 0.40\%           \\
    & 1.0  & 39.20\% & 38.75\% & 37.71\% & 38.70\% & 0.44\%     & 1.49\%       & 0.50\%           \\
    & 2.0  & 39.08\% & 38.56\% & 37.01\% & 38.44\% & 0.52\%     & 2.07\%       & 0.64\%           \\
~   & 3.0  & 39.00\% & 38.43\% & 36.50\% & 38.25\% & 0.57\%     & 2.50\%       & 0.75\%           \\\hdashline
0.3 & 0.25 & 39.51\% & 39.25\% & 38.82\% & 39.24\% & 0.26\%     & 0.68\%       & 0.27\%           \\
    & 0.5  & 39.25\% & 38.86\% & 38.14\% & 38.83\% & 0.39\%     & 1.11\%       & 0.42\%           \\
    & 1.0  & 38.87\% & 38.26\% & 37.07\% & 38.19\% & 0.61\%     & 1.80\%       & 0.68\%           \\
    & 2.0  & 38.30\% & 37.35\% & 35.45\% & 37.17\% & 0.95\%     & 2.85\%       & 1.13\%           \\
~   & 3.0  & 37.85\% & 36.62\% & 34.17\% & 36.32\% & 1.23\%     & 3.68\%       & 1.53\%           \\\hdashline
0.5 & 0.25 & 39.74\% & 39.61\% & 39.30\% & 39.60\% & 0.13\%     & 0.44\%       & 0.13\%           \\
    & 0.5  & 39.48\% & 39.21\% & 38.62\% & 39.20\% & 0.26\%     & 0.86\%       & 0.28\%           \\
    & 1.0  & 38.96\% & 38.42\% & 37.31\% & 38.36\% & 0.54\%     & 1.65\%       & 0.60\%           \\
    & 2.0  & 37.97\% & 36.85\% & 34.92\% & 36.65\% & 1.11\%     & 3.05\%       & 1.31\%           \\
~   & 3.0  & 37.02\% & 35.35\% & 32.82\% & 35.00\% & 1.68\%     & 4.21\%       & 2.03\%           \\\hdashline
0.7 & 0.25 & 39.87\% & 39.81\% & 39.60\% & 39.81\% & 0.06\%     & 0.27\%       & 0.06\%           \\
    & 0.5  & 39.66\% & 39.50\% & 39.02\% & 39.49\% & 0.16\%     & 0.64\%       & 0.17\%           \\
    & 1.0  & 39.11\% & 38.67\% & 37.66\% & 38.61\% & 0.44\%     & 1.45\%       & 0.50\%           \\
    & 2.0  & 37.74\% & 36.52\% & 34.61\% & 36.32\% & 1.22\%     & 3.13\%       & 1.42\%           \\
~   & 3.0  & 36.20\% & 34.10\% & 31.65\% & 33.75\% & 2.10\%     & 4.55\%       & 2.45\%           \\\hdashline
0.9 & 0.25 & 39.94\% & 39.91\% & 39.76\% & 39.91\% & 0.03\%     & 0.17\%       & 0.03\%           \\
    & 0.5  & 39.78\% & 39.68\% & 39.31\% & 39.67\% & 0.10\%     & 0.47\%       & 0.11\%           \\
    & 1.0  & 39.25\% & 38.88\% & 37.97\% & 38.84\% & 0.36\%     & 1.27\%       & 0.41\%           \\
    & 2.0  & 37.52\% & 36.21\% & 34.36\% & 36.01\% & 1.31\%     & 3.16\%       & 1.51\%           \\
~   & 3.0  & 35.27\% & 32.79\% & 30.52\% & 32.46\% & 2.48\%     & 4.75\%       & 2.82\%          \\\hline
\end{tabularx}
\end{table}

\begin{table}[H]
\caption{Zero-vanna approximation for volatility swaps ($\rho=0$, $\sigma_0=20\%$, $\alpha=2$)}
\label{table:t5}
\newcolumntype{Y}{>{\centering\arraybackslash}X}
\newcolumntype{Z}{>{\raggedleft\arraybackslash}X}
\begin{tabularx}{\linewidth}{XXYYYYY} \hline
$H$  &  $T$    & VS   & IV($\hat{k}$)      & ATMI    & VS - IV($\hat{k}$) & VS - ATMI  \\ \hline\hline
0.1 & 0.25 & 18.10\% & 18.10\% & 18.09\% & 0.01\%     & 0.01\%        \\
    & 0.5  & 17.84\% & 17.83\% & 17.82\% & 0.00\%     & 0.01\%        \\
    & 1.0  & 17.53\% & 17.52\% & 17.50\% & 0.01\%     & 0.03\%        \\
    & 2.0  & 17.18\% & 17.15\% & 17.12\% & 0.03\%     & 0.06\%        \\
~   & 3.0  & 16.96\% & 16.91\% & 16.85\% & 0.05\%     & 0.10\%        \\\hdashline
0.3 & 0.25 & 18.52\% & 18.52\% & 18.51\% & 0.01\%     & 0.01\%        \\
    & 0.5  & 17.81\% & 17.80\% & 17.80\% & 0.00\%     & 0.01\%        \\
    & 1.0  & 16.79\% & 16.77\% & 16.75\% & 0.02\%     & 0.04\%        \\
    & 2.0  & 15.41\% & 15.32\% & 15.28\% & 0.09\%     & 0.13\%        \\
~   & 3.0  & 14.41\% & 14.24\% & 14.18\% & 0.17\%     & 0.22\%        \\\hdashline
0.5 & 0.25 & 19.21\% & 19.20\% & 19.20\% & 0.00\%     & 0.01\%        \\
    & 0.5  & 18.46\% & 18.46\% & 18.45\% & 0.00\%     & 0.01\%        \\
    & 1.0  & 17.12\% & 17.10\% & 17.09\% & 0.02\%     & 0.03\%        \\
    & 2.0  & 14.94\% & 14.82\% & 14.79\% & 0.12\%     & 0.15\%        \\
~   & 3.0  & 13.28\% & 13.04\% & 13.01\% & 0.23\%     & 0.27\%        \\\hdashline
0.7 & 0.25 & 19.60\% & 19.60\% & 19.60\% & 0.00\%     & 0.00\%        \\
    & 0.5  & 18.98\% & 18.98\% & 18.98\% & 0.00\%     & 0.00\%        \\
    & 1.0  & 17.54\% & 17.53\% & 17.51\% & 0.02\%     & 0.03\%        \\
    & 2.0  & 14.76\% & 14.62\% & 14.60\% & 0.14\%     & 0.17\%        \\
~   & 3.0  & 12.61\% & 12.37\% & 12.34\% & 0.24\%     & 0.27\%        \\\hdashline
0.9 & 0.25 & 19.81\% & 19.80\% & 19.80\% & 0.00\%     & 0.00\%        \\
    & 0.5  & 19.34\% & 19.34\% & 19.34\% & 0.00\%     & 0.00\%        \\
    & 1.0  & 17.91\% & 17.90\% & 17.89\% & 0.01\%     & 0.02\%        \\
    & 2.0  & 14.65\% & 14.50\% & 14.48\% & 0.15\%     & 0.17\%        \\
~   & 3.0  & 12.20\% & 11.99\% & 11.97\% & 0.21\%     & 0.23\%       \\\hline
\end{tabularx}
\end{table}

\begin{table}[H]
\caption{Zero-vanna approximation for volatility swaps ($\rho=-0.8$, $\sigma_0=20\%$, $\alpha=2$)}
\label{table:t6}
\newcolumntype{Y}{>{\centering\arraybackslash}X}
\newcolumntype{Z}{>{\raggedleft\arraybackslash}X}
\begin{tabularx}{\linewidth}{XXccccccc} \hline
$H$   & $T$    & VS   &  IV($\hat{k}$)     & ATMI    & AS(4.8) & VS -  IV($\hat{k}$) & VS - ATMI & VS - AS(4.8)  \\ \hline\hline
0.1 & 0.25 & 18.10\% & 17.40\% & 17.19\% & 17.39\% & 0.71\%     & 0.91\%       & 0.71\%           \\
    & 0.5  & 17.84\% & 17.04\% & 16.75\% & 17.03\% & 0.80\%     & 1.09\%       & 0.81\%           \\
    & 1.0  & 17.53\% & 16.63\% & 16.22\% & 16.61\% & 0.90\%     & 1.31\%       & 0.92\%           \\
    & 2.0  & 17.18\% & 16.17\% & 15.60\% & 16.13\% & 1.01\%     & 1.58\%       & 1.05\%           \\
~   & 3.0  & 16.96\% & 15.86\% & 15.17\% & 15.80\% & 1.09\%     & 1.78\%       & 1.16\%           \\\hdashline
0.3 & 0.25 & 18.52\% & 17.83\% & 17.63\% & 17.83\% & 0.69\%     & 0.90\%       & 0.69\%           \\
    & 0.5  & 17.81\% & 16.83\% & 16.52\% & 16.82\% & 0.98\%     & 1.29\%       & 0.99\%           \\
    & 1.0  & 16.79\% & 15.45\% & 15.02\% & 15.43\% & 1.34\%     & 1.77\%       & 1.37\%           \\
    & 2.0  & 15.41\% & 13.65\% & 13.12\% & 13.61\% & 1.75\%     & 2.29\%       & 1.79\%           \\
~   & 3.0  & 14.41\% & 12.43\% & 11.85\% & 12.37\% & 1.98\%     & 2.55\%       & 2.03\%           \\\hdashline
0.5 & 0.25 & 19.21\% & 18.83\% & 18.67\% & 18.83\% & 0.37\%     & 0.54\%       & 0.38\%           \\
    & 0.5  & 18.46\% & 17.76\% & 17.47\% & 17.75\% & 0.70\%     & 0.99\%       & 0.71\%           \\
    & 1.0  & 17.12\% & 15.89\% & 15.47\% & 15.86\% & 1.24\%     & 1.65\%       & 1.26\%           \\
    & 2.0  & 14.94\% & 13.10\% & 12.63\% & 13.06\% & 1.85\%     & 2.32\%       & 1.88\%           \\
~   & 3.0  & 13.28\% & 11.21\% & 10.77\% & 11.18\% & 2.07\%     & 2.51\%       & 2.10\%           \\\hdashline
0.7 & 0.25 & 19.60\% & 19.42\% & 19.30\% & 19.42\% & 0.18\%     & 0.30\%       & 0.18\%           \\
    & 0.5  & 18.98\% & 18.53\% & 18.28\% & 18.52\% & 0.46\%     & 0.70\%       & 0.46\%           \\
    & 1.0  & 17.54\% & 16.50\% & 16.11\% & 16.48\% & 1.05\%     & 1.44\%       & 1.07\%           \\
    & 2.0  & 14.76\% & 13.01\% & 12.59\% & 12.98\% & 1.75\%     & 2.17\%       & 1.78\%           \\
~   & 3.0  & 12.61\% & 10.78\% & 10.42\% & 10.75\% & 1.83\%     & 2.19\%       & 1.86\%           \\\hdashline
0.9 & 0.25 & 19.81\% & 19.72\% & 19.63\% & 19.72\% & 0.09\%     & 0.17\%       & 0.09\%           \\
    & 0.5  & 19.34\% & 19.05\% & 18.85\% & 19.05\% & 0.29\%     & 0.49\%       & 0.29\%           \\
    & 1.0  & 17.91\% & 17.05\% & 16.68\% & 17.03\% & 0.87\%     & 1.23\%       & 0.89\%           \\
    & 2.0  & 14.65\% & 13.08\% & 12.72\% & 13.06\% & 1.57\%     & 1.94\%       & 1.60\%           \\
~   & 3.0  & 12.20\% & 10.71\% & 10.41\% & 10.69\% & 1.49\%     & 1.79\%       & 1.51\%          \\\hline
\end{tabularx}
\end{table}

As we can see from the tables, IV ($\hat{k}$) approximates the volatility swap strike better than ATMI in all cases.
Also, since the error order in the uncorrelated case is higher than that of the correlated case, 
and the error is always multiplied by the correlation, 
IV ($\hat{k}$) is a more accurate approximation of the volatility swap strike when the correlation is small.
While the values of AS(4.8) are close to those of IV($\hat{k}$), IV($\hat{k}$) is better in our settings.

For long term maturities, smaller Hurst parameters result in smaller errors, while for short term maturities, larger Hurst parameters result in smaller errors.
Also, for the same Hurst parameter the results are more accurate for short maturities.
This is because the order of the error is expressed as $T$ to the power of the Hurst parameter.
Moreover, since $D_s \sigma_t = \frac12 \exp(\frac12 \alpha (W^H_t - W^H_s) - \frac{1}{4}\alpha^2(t^{2H} - s^{2H}))\alpha \sigma_t {\bf 1}_{\{s\le t\}}$, the errors also depend on the size of $\sigma_0$ and $\alpha$.
Thus, as shown in Theorems \ref{uncorrelated}, \ref{themaintheorem} and Tables \ref{table:t3} - \ref{table:t6}, the errors are proportional to $\sigma_0$ and $\alpha$.

The numerical results support our analytical results, and moreover tell us that in the uncorrelated case the zero vanna approximation is highly accurate and for most practical purposes can be taken as the price of the volatility swap. As volatility or volatility of volatility moves away from zero, although the approximation cannot anymore serve as the price of the volatility swap, it can still be used as a benchmark price against which to validate model specific prices, or even as an indicative bid-price under normal market conditions. This is especially true in foreign exchange (FX) markets where volatility or volatility of volatility values are typically not as extreme as in equity markets.

Furthermore, since the zero vanna approximation can be calculated as frequently as the implied volatility smile is updated, and without added computational costs, it is possible to use the zero vanna implied volatility to monitor market conditions for volatility swaps prices as well as using it as a risk management tool. Deducing volatility swap benchmark prices directly from the observable implied volatility smile is clearly much faster than evaluating an exact model-specific volatility swap strike. The latter will require Monte Carlo or PDE methods as there are almost no models that yield the exact volatility swap strike in closed form. Even when expensive Monte Carlo or PDE methods are used to calculate the `exact' price, numerical and calibration errors cannot be avoided.

\section{Conclusion}
By using techniques from Malliavin calculus we have extended the validity of the zero vanna implied volatility as an approximation for pricing volatility swaps to fractional stochastic volatility models. 
Furthermore, we have proved that even though the zero vanna approximation for the volatility swap strike is accurate for zero correlation and for all values of the Hurst parameter, it is not exact. 
Thus, indirectly it is proved that the Rolloos-Arslan approximation is not equivalent to the Carr-Lee approximation for volatility swaps as the latter is exact for zero correlation. 
However, in the uncorrelated case and for most practical purposes it can be treated as exact. 

It has also been shown that the zero vanna approximation has a higher or equal rate of convergence than the Al\`{o}s and Shiraya (2019) model-free result and the ATMI approximation. In the uncorrelated case, the zero vanna approximation has a higher order than the ATM implied volatility for all values of $H$. When correlation deviates from zero a comparison of the order of convergence is more nuanced: For short maturities and Hurst value $H > 1/2$ the zero vanna implied volatility converges faster to the exact volatility swap price. When $H=1/2$ the order on time to maturity $T-t$ is the same as for the ATM implied volatility, however the first order of correlation $\rho$ is not present in the leading terms of the zero vanna approximation. If $H<1/2$ the leading terms of the zero vanna approximation is the same as of the ATM implied volatility, but the order on $T-t$ of the first order of $\rho$ is higher than that of the ATM implied volatility.

\appendix
\section{Malliavin calculus}\label{appendix1}
In this appendix, we present the basic Malliavin calculus results we use in this paper. The first one is the Clark-Ocone formula, that allows us to compute explicitly the martingale representation of a random variable $F\in \mathbb{D}^{1,2}_W$.

\begin{theorem}{\it Clark-Ocone formula} Consider a Brownian motion $W=\{W_t, t\in [0,T]\}$ and a random variable $F\in \mathbb{D}^{1,2}_W$. Then
$$
F=E[F]+\int_0^T E_r[D_r^WF] dW_r.
$$
\end{theorem}

We will also make use of the following anticipating It\^o's formula (adapted from Nualart and Pardoux (1998)), that allows us to work with non-adapted processes.

\begin{proposition}
\label{Ito2} Consider a process $X$ of the form $X_t = X_0+\int_0^t u_sdW_s+\int_0^t u'_s dB_s+\int_0^t v_sds$,
where $X_0$ is a constant and $u, v$ are square-integrable  stochastic processes adapted to the filtration generated by $
W $ and $B$. Consider also a process $Y_t=\int_t^T \theta_s ds$, for some $\theta\in\mathbb{L}^{1,2}_{W}$ and adapted to the filtration generated by W.  Let
$F:[0,T]\times \mathbb{R}^{2}\rightarrow \mathbb{R}$ be a function
in $C^{1,2} ([0,T]\times \mathbb{R}^{2})$ such that there exists a
positive constant $C$ such that, for all $t\in \left[ 0,T\right]
,$ $F$ and its partial derivatives evaluated in $\left(
t,X_{t},Y_{t}\right)$ are bounded by $C.$ Then it follows that
\begin{eqnarray}
F(t,X_{t},Y_{t}) &=&F(0,X_{0},Y_{0})+\int_{0}^{t}{\partial _{s}F}%
(s,X_{s},Y_{s})ds \nonumber\\
&&+\int_{0}^{t}{\partial _{x}F}(s,X_{s},Y_{s}) v_sds \nonumber\\
&&+\int_{0}^{t}{\partial _{x}F}(s,X_{s},Y_{s})(u _{s}dW_{s}+u_s'dB_s) \nonumber\\
&&-\int_{0}^{t}{\partial _{y}F}(s,X_{s},Y_{s})\theta
_{s}^{2}ds+
\int_{0}^{t}{\partial _{xy}^{2}F}(s,X_{s},Y_{s})D^-Y_s u_s ds \nonumber\\
&&+\frac{1}{2}\int_{0}^{t}{\partial
_{xx}^{2}F}(s,X_{s},Y_{s})(u_s^2+(u'_s)^2)ds, 
\end{eqnarray} 
where  $D^-Y_s:=\int_s^T D_s \theta_r dr$.
\end{proposition}
\begin{remark}
\label{ext}
The anticipating It\^o's formula also holds for processes of the form $F(s,X_{s},Y^1_{s},...,Y^n_s)$, where $Y^i=\int_t^T\theta^i_sds$, $i=i,...,n$ , just replacing the term
\begin{equation}
\label{extra}
\int_{0}^{t}{\partial _{xy}^{2}F}(s,X_{s},Y_{s})D^-Y _{s}u_sds,
\end{equation}
by 
$$
\sum_i^n\int_{0}^{t}{\partial _{xy^i}^{2}F}(s,X_{s},Y^1_{s},...,Y^n_s)D^-Y^i _{s}ds.
$$
\end{remark}

\section{Proofs}\label{appendix2}

This section shows the proofs of Propositions and Theorems in Section \ref{sec3}.
Firstly, we give some Greeks of Black-Scholes formula.

A direct calculation gives us that  $k\in\mathbb{R}$ and all $u>0$:
$$
(BS^{-1})'(k,u)=\frac{1}{\frac{\partial BS}{\partial\sigma}(k,BS^{-1}(k,u))}.
$$
Then it follows that 
\begin{eqnarray}
(BS^{-1})''(k,u)&=&-\frac{1}{(\frac{\partial BS}{\partial\sigma}(k,BS^{-1}(k,u)))^2}\frac{\partial^2 BS}{\partial\sigma^2}(k,BS^{-1}(k,u))\frac{1}{\frac{\partial BS}{\partial\sigma}(k,BS^{-1}(k,u))}\nonumber\\
&=&-\frac{1}{(\frac{\partial BS}{\partial\sigma}(k,BS^{-1}(k,u)))^3}\frac{\partial^2 BS}{\partial\sigma^2}(k,BS^{-1}(k,u)).
\end{eqnarray}
Now, the classical relationship between the {\it Vomma} and the {\it Vega}
$$
\frac{\partial^2 BS}{\partial\sigma^2}(k,\sigma)=\frac{\partial BS}{\partial\sigma}(k,\sigma)\frac{d_1(k,\sigma)d_2(k,\sigma)}{\sigma},
$$
allows us to write
$$
(BS^{-1})''(k,u)=\frac{1}{(\frac{\partial BS}{\partial\sigma}(k,BS^{-1}(k,u)))^2}\frac{(BS^{-1}(k,u))^4(T-t)^2-4(X_t-k)^2}{4(BS^{-1}(k,u))^3(T-t)}.
$$
Finally, as
$$\frac{\partial BS}{\partial\sigma}(k,BS^{-1}(k,u))=
\exp(X_{t})N^{\prime }(d_1\left( k,BS^{-1}(k,u)\right))\sqrt{T-t},$$ the above equality reduces to
\begin{equation}
\label{greek}
(BS^{-1})''(k,u)=\frac{(BS^{-1}(k,u))^4(T-t)^2-4(X_t-k)^2}{4(\exp(X_{t})N^{\prime }(d_1\left( k,BS^{-1}(k,u)\right))(T-t))^2(BS^{-1}(k,u))^3}.
\end{equation}

By using these formula, we show the proofs.

\begin{proof}[Proof of Proposition \ref{General}]
This proof is decomposed into several steps.\\

\noindent {\it Step 1} First, we will show that
\begin{equation}
\label{primerpas}
I\left( t,T,X_t,\hat{k}_t\right)  =E_{t}\left[ v_{t}\right] +\frac{1}{2} E_t
\left[\int_{t}^{T} \left( BS^{-1}\right) ^{\prime \prime}\left( \hat{k}_t, \Lambda_r\right)U_{r}^{2}dr\right].
\end{equation}
Observe that, as $\rho=0$, the Hull and White formula gives us that $V_t=\Lambda_t$. Then, as in the proof of Proposition 3.1 in Al\`os and Shiraya (2019) we can write
\begin{equation}
\label{impliedmartingale}
I(t,T,X_t,\hat{k}_t)=BS^{-1}(\hat{k}_t,\Lambda_t)=E_t[BS^{-1}(\hat{k}_t,\Lambda_t)].
\end{equation}
Now, (H2) and the Clark-Ocone formula (see Appendix \ref{appendix1}) give us that $\Lambda$ admits the martingale representation given by
\begin{eqnarray}
d\Lambda_r&=&E_r[D_r^W(BS(t,T,X_t,\hat{k}_t,v_t)]\nonumber\\
&=&E_r\left[\frac{\partial BS}{\partial \sigma}(t,T,X_t,\hat{k}_t,v_t)\frac{1}{2v_t(T-t)}\int_r^TD_r^W\sigma_s^2 ds\right]dW_r \nonumber\\
&=&U_rdW_r.
\end{eqnarray}
Then, a direct application of the classical It\^o's formula gives us that, after taking expectations:
\begin{eqnarray}
\label{Itoimplied}
E_t[BS^{-1}(\hat{k}_t,\Lambda_t)]
=E_t[BS^{-1}(\hat{k}_t,\Lambda_T)]-\frac12 E_t \left[ \int_{t}^{T} \left( BS^{-1}\right) ^{\prime \prime}\left( \hat{k}_t, \Lambda_r\right)d\langle \Lambda,\Lambda\rangle_r\right].
\end{eqnarray}
Now, as $\Lambda_T=BS\left( t,T,X_t,\hat{k}_t,v_{t}\right)$, (\ref{impliedmartingale}) and (\ref{Itoimplied}) imply that
\begin{eqnarray}
\label{Itoimplied2}
I(t,T,X_t,\hat{k}_t)=E_{t}\left[ v_{t}\right] -\frac12 E_t \left[ \int_{t}^{T} \left( BS^{-1}\right) ^{\prime \prime}\left( \hat{k}_t, \Lambda_r\right)d\langle \Lambda,\Lambda\rangle_r\right].
\end{eqnarray}
That is, $$
I(t,T,X_t,\hat{k}_t)=E_{t}\left[ v_{t}\right] -\frac12 E_t \left[ \int_{t}^{T} \left( BS^{-1}\right) ^{\prime \prime}\left( \hat{k}_t, \Lambda_r\right)U_r^2dr\right].
$$
{\it Step 2} Next, let us see that 
\begin{eqnarray}
\label{secondstep}
\lefteqn{E_t \left[\int_{t}^{T} \left( BS^{-1}\right) ^{\prime \prime}\left(\hat{k}_t, \Lambda_r\right)U_{r}^{2}dr\right]}\nonumber\\
&=&E_t\Bigg[\int_{t}^{T}\left( BS^{-1}\left( \hat{k}_t, \Lambda_r\right)\right) ^{\prime \prime\prime} (D^-A)_r U_r dr\Bigg]\nonumber\\
&&+\frac12E_t\Bigg[
\int_{t}^{T}\left( BS^{-1}\left( \hat{k}_t, \Lambda_r\right)\right) ^{(iv)} A_r U_r^2 dr \Bigg].
\label{I-E[v]}
\end{eqnarray}
Towards this end, we apply  the anticipating It\^o's formula (see Appendix \ref{appendix1}) to the process
$$
 \left( BS^{-1}\right) ^{\prime \prime}\left( \hat{k}_t, \Lambda_r\right)A_r,
$$
where we take $F(X_r,Y_r)=\left(BS^{-1}\right) ^{\prime \prime}(\hat{k}_t,X_r)Y_r$,  with $X_r=\Lambda_r$ (then $u_r=U_r$ and $u'_r=0$), and $Y_r=A_r=-\frac12 \int_r^T U_\theta^2 d\theta$. Notice that, because of (H1) and (H2), $F$ and its derivatives evaluated at $(X_r,Y_r)$ are bounded. Then, taking into account that $dA_r=-\frac12 U_r^2dr$ and applying  Proposition \ref{Ito2} we get, after taking expectations
\begin{eqnarray}
E_t \bigg[
\left( BS^{-1}\right) ^{\prime \prime}\left(\hat{k}_t, \Lambda_T\right)A_T\bigg]
&=&E_t \bigg[\left( BS^{-1}\right) ^{\prime \prime}\left( \hat{k}_t, \Lambda_t\right)A_t \bigg]\nonumber\\
&& -\frac{1}{2} E_t\Bigg[
\int_{t}^{T} \left( BS^{-1}\right) ^{\prime \prime}\left(\hat{k}_t, \Lambda_r\right)U_r^2 dr\Bigg]\nonumber\\
&&+ E_t\Bigg[\int_{t}^{T}\left( BS^{-1}\left( \hat{k}_t, \Lambda_r\right)\right) ^{\prime \prime\prime}(D^-A)_r U_r dr\Bigg]\nonumber\\
&&+\frac12 E_t\Bigg[
\int_{t}^{T}\left( BS^{-1}\left( \hat{k}_t, \Lambda_r\right)\right) ^{(iv)} A_r U_r^2 dr \Bigg].
\end{eqnarray}

Equality (\ref{greek}) and the fact that $\Theta_{r}(k):=BS^{-1}(k,\Lambda_r)$  give us that 
\begin{eqnarray}
\left( BS^{-1}\right) ^{\prime \prime }\left(\hat{k}_t,{\Lambda _{r}}%
\right)
&=&\frac{(\Theta_r(\hat{k}_t))^4(T-t)^2-4(X_t-\hat{k}_t)^2}{4\left( \exp(X_{t})N^{\prime }(d_1\left(\hat{k}_t,\Theta_r(\hat{k}_t)\right))(T-t) \right)^{2} (\Theta_r(\hat{k}_t))^3}.\nonumber
\end{eqnarray}
In particular, 
$\left( BS^{-1}\right) ^{\prime \prime }\left(\hat{k}_t,{\Lambda _{t}}\right)=0$ and
\begin{eqnarray}
\left( BS^{-1}\right) ^{\prime \prime }\left(\hat{k}_t,{\Lambda _{T}}\right)
=\frac{(v_t^4-(\Theta_t(\hat{k}_t))^4)}{4\left( \exp(X_{t})N^{\prime }(d_1\left( k_t,v_t\right)) \right)^{2} v_t^3},\label{bs-1''}
\end{eqnarray}
which implies that $\left( BS^{-1}\right) ^{\prime \prime }\left(\hat{k}_t,{\Lambda _{T}}\right)A_T = 0$. Then
\begin{eqnarray}
\frac{1}{2} E_t \left[
\int_{t}^{T} \left( BS^{-1}\right) ^{\prime \prime}\left(\hat{k}_t, \Lambda_r\right)U_r^2dr\right]
&=& E_t\Bigg[\int_{t}^{T}\left( BS^{-1}\left( \hat{k}_t, \Lambda_r\right)\right) ^{\prime \prime\prime}(D^-A)_r U_r dr\Bigg]\nonumber\\
&&+\frac12 E_t\Bigg[
\int_{t}^{T}\left( BS^{-1}\left( \hat{k}_t, \Lambda_r\right)\right) ^{(iv)} A_r U_r^2 dr \Bigg],
\end{eqnarray}
which completes the proof.
\end{proof}

\begin{proof}[Proof of Theorem \ref{uncorrelated}]
Again, the proof is decomposed into several steps.

\noindent{\it Step 1} We start by showing that
\begin{eqnarray}
\label{impliedexpansion}
I\left( t,T,X_t,\hat{k}_t\right) &=& E_{t}\left[ v_{t}\right] \nonumber\\
&&+\left( BS^{-1}\left( \hat{k}_t, \Lambda_t\right)\right) ^{\prime \prime\prime} E_t\Bigg[\int_{t}^{T}(D^-A)_r U_r dr\Bigg]\nonumber\\
&&+ \frac12 \left( BS^{-1}\left( \hat{k}_t, \Lambda_t\right)\right) ^{(iv)} E_t\Bigg[
\int_{t}^{T} A_r U_r^2 dr \Bigg]\nonumber\\
&&+T_1+T_2+T_3+T_4,
\end{eqnarray}
where
\begin{eqnarray*}
T_1&=&E_t\Bigg[\int_{t}^{T}\left( BS^{-1}\left( \hat{k}_t, \Lambda_r\right)\right) ^{(iv)}(D^-\Psi)_r U_r dr\Bigg],
\\
T_2&=& \frac12 E_t\Bigg[\int_{t}^{T}\left( BS^{-1}\left( \hat{k}_t, \Lambda_r\right)\right) ^{(v)}\Psi_r U_r^2 dr\Bigg],
\\
T_3&=&\frac{1}{2} E_t\Bigg[\int_{t}^{T}\left( BS^{-1}\left( \hat{k}_t, \Lambda_r\right)\right) ^{(v)}(D^-\Phi)_r U_r dr\Bigg],
\end{eqnarray*}
and
\begin{eqnarray*}
T_4&=& \frac{1}{4} E_t\Bigg[\int_{t}^{T}\left( BS^{-1}\left( \hat{k}_t, \Lambda_r\right)\right) ^{(vi)}\Phi_r U_r^2 dr\Bigg],
\end{eqnarray*}
with $\Psi_t:=\int_{t}^{T}(D^-A)_r U_r dr$ and $\Phi_t:=\int_{t}^{T} A_r U_r^2 dr$.
Towards this end we can apply the anticipating It\^o's formula (see Remark \ref{ext}) to the processes
$$
\left( BS^{-1}\left( \hat{k}_t, \Lambda_s\right)\right) ^{\prime \prime\prime} \int_{s}^{T}(D^-A)_r U_r dr=:\left( BS^{-1}\left( \hat{k}_t, \Lambda_s\right)\right) ^{\prime \prime\prime} \Psi_s,
$$
and 
$$
\frac14 \left( BS^{-1}\left( \hat{k}_t, \Lambda_s\right)\right) ^{(iv)}
\int_{s}^{T} A_r U_r^2 dr =:\frac14 \left( BS^{-1}\left( \hat{k}_t, \Lambda_s\right)\right) ^{(iv)}\Phi_s.
$$
Then, the same arguments as in the proof of  Proposition \ref{General} allow us to write
\begin{eqnarray}
\lefteqn{E_t\Bigg[\int_{t}^{T}\left( BS^{-1}\left( \hat{k}_t, \Lambda_r\right)\right) ^{\prime \prime\prime}(D^-A)_r U_r dr\Bigg]}\nonumber\\
&=&\left( BS^{-1}\left( \hat{k}_t, \Lambda_t\right)\right) ^{\prime \prime\prime} E_t\Bigg[\int_t^T(D^-A)_r U_r dr\Bigg]\nonumber\\
&&+ E_t\Bigg[\int_{t}^{T}\left( BS^{-1}\left( \hat{k}_t, \Lambda_r\right)\right) ^{(iv)}(D^-\Psi)_r U_r dr\Bigg]\nonumber\\
&&+ \frac12 E_t\Bigg[\int_{t}^{T}\left( BS^{-1}\left( \hat{k}_t, \Lambda_r\right)\right) ^{(v)}\Psi_r U_r^2 dr\Bigg]\nonumber\\
&=&\left( BS^{-1}\left( \hat{k}_t, \Lambda_t\right)\right) ^{\prime \prime\prime} E_t\Bigg[\int_t^T(D^-A)_r U_r dr\Bigg]+T_1+T_2,
\end{eqnarray}
and
\begin{eqnarray}
\lefteqn{E_t\Bigg[
\int_{t}^{T}\left( BS^{-1}\left( \hat{k}_t, \Lambda_r\right)\right) ^{(iv)} A_r U_r^2 dr \Bigg]}\nonumber\\
&=&\left( BS^{-1}\left( \hat{k}_t, \Lambda_t\right)\right) ^{(iv)} E_t\Bigg[
\int_{t}^{T}A_r U_r^2 dr \Bigg]\nonumber\\
&&+E_t\Bigg[
\int_{t}^{T}\left( BS^{-1}\left( \hat{k}_t, \Lambda_r\right)\right) ^{(v)} (D^-\Phi)_r U_r dr \Bigg]\nonumber\\
&&+\frac12 E_t\Bigg[\int_{t}^{T}\left( BS^{-1}\left( \hat{k}_t, \Lambda_r\right)\right) ^{(vi)}\Phi_r U_r^2 dr\Bigg]\nonumber\\
&=&\left( BS^{-1}\left( \hat{k}_t, \Lambda_t\right)\right) ^{(iv)} E_t\Bigg[
\int_{t}^{T}A_r U_r^2 dr \Bigg]+T_3+T_4.
\end{eqnarray}
{\it Step 2} Now, let us study the term
$$
\left( BS^{-1}\left( \hat{k}_t, \Lambda_t\right)\right) ^{\prime \prime\prime} E_t\Bigg[\int_t^T(D^-A)_r U_r dr\Bigg].
$$
On one hand,
\begin{eqnarray}
\label{terceraderivada}
\left( BS^{-1}\left( \hat{k}_t, \Lambda_t\right)\right) ^{\prime \prime\prime}
&=&\frac{-\frac{\partial^3 BS}{\partial \sigma^3}(\hat{k}_t, \Theta_t(\hat{k}_t))\left(\frac{\partial BS}{\partial \sigma}(\hat{k}_t, \Theta_t(\hat{k}_t))\right)^3
+3\left(\frac{\partial^2 BS}{\partial \sigma^2}(\hat{k}_t, \Theta_t(\hat{k}_t))\right)^2\left(\frac{\partial BS}{\partial \sigma}(\hat{k}_t, \Theta_t(\hat{k}_t))\right)^2
}{\left(\frac{\partial BS}{\partial \sigma}(\hat{k}_t, \Theta_t(\hat{k}_t))\right)^7}\nonumber\\
&=&\frac{-\frac{\partial^3 BS}{\partial \sigma^3}(\hat{k}_t, \Theta_t(\hat{k}_t))}{\left(\frac{\partial BS}{\partial \sigma}(\hat{k}_t, \Theta_t(\hat{k}_t))\right)^4} 
+o\left((T-t)^{-\frac{1}{2}}\right)\nonumber\\
&=&(2\pi)^{\frac32}\exp\left(-3X_t+\frac{3}{2}(\Theta_t(\hat{k}_t))^2(T-t)\right)(T-t)^{-\frac12}+o\left((T-t)^{-\frac{1}{2}}\right).
\end{eqnarray}
On the other hand,
\begin{eqnarray}
\label{DA}
(D^-A)_r=\frac{1}{2}\int_r^T D_r^WU_s^2ds=\int_r^TU_s D_r^W U_s ds.
\end{eqnarray}
The vega-delta-gamma relationship 
\begin{equation}
\frac{\partial BS}{\partial \sigma}(t,T,x,k,\sigma)\frac{1}{\sigma (T-t)}=\left(\frac{\partial }{\partial x^2}-\frac{\partial }{\partial x}\right) BS(t,T,x,k,\sigma),
\end{equation}
allows us to write
\begin{eqnarray}
\label{U}
U_s&=&E_s\left[\frac{\partial BS}{\partial \sigma}(t,T,X_t,\hat{k}_t,v_t)\frac{1}{2v_t(T-t)}\int_s^T D_s^W \sigma_u^2 du\right]  \nonumber\\
&=&\frac12    E_s\left[G(t,T,X_t,\hat{k}_t,v_t)\int_s^T D_s^W \sigma_u^2 du\right],
\end{eqnarray}
and
\begin{eqnarray}
\label{DU}
D_r^WU_s&=& E_s\Bigg[ \frac12 
G(t,T,X_t,\hat{k}_t,v_t)\left(\frac{d_1(\hat{k}_t,v_t)d_2(\hat{k}_t,v_t)}{2v_t^2(T-t)} - \frac{1}{2v_t^2(T-t)}\right)
\left(\int_s^T D_s^W\sigma_u^2 du\right)\left(\int_r^TD_r^W\sigma_u^2 du\right)
\nonumber\\
&&+\frac12G(t,T,X_t,\hat{k}_t,v_t)\left(\int_s^TD_r^WD_s^W\sigma_u^2 du\right)\Bigg].\nonumber  
\end{eqnarray}
Then, from the equation for $G$ and (H2') we can deduce that 
\begin{eqnarray}
\label{DA2}
(D^-A)_r
&=&
\frac{1}{4}
\int_r^T E_s\left[\frac{e^{\hat{k}_t}N'(d_2(\hat{k},v_t))}{v_t\sqrt{T-t}}\int_s^TD_s^W\sigma_u^2 du    \right] \nonumber\\
&&\times E_s \Bigg[
\frac{e^{\hat{k}_t}N'(d_2(\hat{k},v_t))}{v_t\sqrt{T-t}}
\left(
\frac{-1}{2 v_t^2(T-t)}
\left(\int_s^TD_s^W\sigma_u^2 du\right)\left(\int_r^TD_r^W\sigma_u^2 du\right)+\left(\int_s^TD_r^WD_s^W\sigma_u^2 du\right)\right)
\Bigg] ds\nonumber\\
&&
+ O(\nu^3 (T-t)^{3H+\frac{1}{2}}),
\end{eqnarray}
which implies that
\begin{eqnarray}
\lefteqn{E_t\Bigg[\int_t^T(D^-A)_r U_r dr\Bigg]}\nonumber\\
&=& \frac{1}{8} E_t\Biggg[\int_t^T\frac{e^{\hat{k}_t}N'(d_2(\hat{k},v_t))}{v_t\sqrt{T-t}} \left(\int_r^TD_r^W\sigma_u^2 du\right)
\int_r^T  E_s\left[\frac{e^{\hat{k}_t}N'(d_2(\hat{k},v_t))}{v_t\sqrt{T-t}}\int_s^TD_s^W\sigma_u^2 du\right]\nonumber\\
&&\times E_s\left[\frac{e^{\hat{k}_t}N'(d_2(\hat{k},v_t))}{v_t\sqrt{T-t}} \left(\frac{-1}{2v_t^2(T-t)}\left(\int_s^TD_s^W\sigma_u^2 du\right)\left(\int_r^TD_r^W\sigma_u^2 du\right) + 
\left(\int_s^TD_r^WD_s^W\sigma_u^2 du  \right) \right)\right]dsdr\Biggg]\nonumber\\
&&+ o(\nu^4(T-r)^{4H+1}).
\end{eqnarray}
This, jointly with (\ref{terceraderivada}) and (H2') shows
$$
\left( BS^{-1}\left( \hat{k}_t, \Lambda_t\right)\right) ^{\prime \prime\prime} E_t\Bigg[\int_t^T(D^-A)_r U_r dr\Bigg]=O(\nu^4(T-t)^{4H+1}).
$$

{\it Step 3} In order to calculate the term
$$
\frac14 \left( BS^{-1}\left( \hat{k}_t, \Lambda_r\right)\right) ^{(iv)} E_t\Bigg[
\int_{t}^{T} A_r U_r^2 dr \Bigg].
$$
Note that
\begin{equation}
\label{derivada4}
\left( BS^{-1}\left( \hat{k}_t, \Lambda_t\right)\right) ^{(iv)}=-\frac{3(2\pi)^{2}}{\Theta_t(\hat{k}_t)}\exp\left(-4X_t+2(\Theta_t(\hat{k}_t))^2(T-t)\right)(T-t)^{-1} + o\left((T-t)^{-1}\right).
\end{equation}
On the other hand, 
\begin{eqnarray}
\lefteqn{E_t\Bigg[\int_{t}^{T} A_r U_r^2 dr \Bigg]}\nonumber\\
&=&\frac{1}{2}E_t\Bigg[\int_{t}^{T} \left(\int_r^TU_s^2ds\right) U_r^2 dr \Bigg]\nonumber\\
&=&\frac{1}{4}E_t\Bigg[ \left(\int_t^TU_r^2dr\right)^2  \Bigg]\nonumber\\
&=&\frac{1}{4}E_t\Bigg[ \left(\int_t^T\left(    E_r\left[\frac{\partial BS}{\partial \sigma}(t,T,X_t,\hat{k}_t,v_t)\frac{1}{2v_t(T-t)}\int_r^TD_r^W\sigma_s^2 dr\right]
    \right)^2ds\right)^2  \Bigg].
\end{eqnarray}
Together with (\ref{derivada4}) this gives us 
$$
\left( BS^{-1}\left( \hat{k}_t, \Lambda_t\right)\right) ^{(iv)} E_t\Bigg[
\int_{t}^{T} A_r U_r^2 dr \Bigg]=O(\nu^4(T-t)^{4H+1})
$$

{\it Step 4} Next, let us prove that $T_2+T_4=o(\nu^4 (T-t)^{4H+1})$. The computations in Step 2 and Step 3 prove that
$\Psi_r=O(\nu^4(T-t)^{4H+\frac32})$ and $\Phi_r=O(\nu^4(T-t)^{4H+2})$. Moreover, $U_r=O(\nu(T-t)^H)$ and direct computations give us that
$$
BS^{-1}\left( \hat{k}_t, \Lambda_r\right)^{(v)}\leq C(T-r)^{-\frac32},
$$
and 
$$
BS^{-1}\left( \hat{k}_t, \Lambda_r\right)^{(vi)}\leq C(T-r)^{-2},
$$
for some positive constant $C$. Then, straightforward computations allow us to check that $T_2+T_4=o(\nu^4(T-t)^{4H+1})$. 

{\it Step 5} The final step is to show that $T_1+T_3=o(\nu^4 (T-t)^{4H+1})$. We have that
\begin{eqnarray}
D^-\Psi_t&:=&\int_{t}^{T}D_t^W ((D^-A)_r U_r )dr\nonumber\\
&=&\int_{t}^{T}(D_t^W (D^-A)_r) U_r dr+\int_{t}^{T} (D^-A)_r D_t^W U_r dr,
\end{eqnarray}
and
\begin{eqnarray}
D^-\Phi_t&:=&\int_{t}^{T} D_t^W(A_r U_r^2)dr\nonumber\\
&=&\int_{t}^{T}(D_t^W A_r) U_r^2 dr+2\int_{t}^{T} U_r A_r D_t^W( U_r) dr.
\end{eqnarray}
Notice that $G(t,T,X_t,\hat{k}_t,v_t)\le C_t\sqrt{T-t}$ for some $C_t>0$. Them, from the definition of $U_r$, and (H2), we directly get that
$$U_r=O(\nu(T-t)^H).$$
On the other hand, a direct computation gives us that 
$$G(t,T,X_t,\hat{k}_t,v_t)\left(\frac{d_1(\hat{k}_t,v_t)d_2(\hat{k}_t,v_t)}{2 v_t^2(T-t)} - \frac{1}{2v_t^2(T-t)}\right)<C_t(T-t)^{-\frac32}.$$
Then, (\ref{DU}), (H2) and straightforward computations lead to
$$D_tU_r=O(\nu^2 (T-t)^{2H-\frac12}).$$
Moreover, 
\begin{eqnarray}
&&\frac{\partial}{\partial v_t} \left[G(t,T,X_t,\hat{k}_t,v_t)\left(\frac{d_1(\hat{k}_t,v_t)d_2(\hat{k}_t,v_t)}{2v_t^2(T-t)} - \frac{1}{2v_t^2(T-t)}\right)\right]\frac{1}{2v_t(T-t)}<C_t(T-t)^{-\frac52},
\end{eqnarray}
from where we deduce that
$$D_tD_sU_r=O(\nu^3(T-t)^{3H-1}).$$ 
Then we can easily see that, under (H2'),
\begin{eqnarray}
\label{DA}
D_t^W (D^-A)_r&=&\int_r^T D_t^W (U_s D_r^WU_s)ds\nonumber\\
&=&\int_r^T D_t^WU_s D_r^W U_s ds + \int_r^TU_s(D_t^W D_r^W U_s )ds\nonumber\\
&=&O(\nu^4 (T-r)^{4H}).
\end{eqnarray}
Then we deduce that $D^-\Psi_t=O(\nu^5(T-t)^{5H+1})$ and  $D^-\Phi_t=O(\nu^5(T-t)^{5H+\frac32})$. Again, direct computations allow us to see that for some positive constant $C$,
$$
BS^{-1}\left( \hat{k}_t, \Lambda_r\right)^{(iv)}\leq C(T-r)^{-1},
$$
which allows us to see that $T_1+T_3=o(\nu^4 (T-t)^{4H+1})$.  Now the proof is complete.
\end{proof}

\begin{proof}[Proof of Theorem \ref{themaintheorem}] 
The proof of this result follows similar ideas as the proof Theorem 2 in Al\`os and Shiraya (2019). 
Notice that Proposition \ref{Theoremcorrelatedcase} gives us that
$$
I( t,T,X_t,\hat{k}_t) -E[v_t]=T_1+T_2,
$$
where
\begin{eqnarray*}
T_1&=&I^{0}( t,T,X_t, \hat{k}_t^0)-E[v_t],\\
T_2&=&\frac{\rho }{2}E_t\left[\int_{t}^{T}( BS^{-1}) ^{\prime }( \hat{k}_t,\Gamma _{s}) H(s,T,X_{s},\hat{k}_t,v_{s})\zeta_{s}ds\right].
\end{eqnarray*}
Let us first see that $T_1=O( (T-t)^{2H+1})$. Notice that
$$
T_1=\frac12 E_t \left[ \int_{t}^{T} \left( BS^{-1}\right) ^{\prime \prime}\left( \hat{k}_t, \Lambda_r\right)U_r^2dr\right].
$$
Now, as
\begin{eqnarray}
\left( BS^{-1}\right) ^{\prime \prime }\left(\hat{k}_t,{\Lambda _{r}}\right) 
&=&\frac{(\Theta_r(\hat{k}_t))^4(T-t)^2-4(X_t-k_t)^2}{4\left( \exp(X_{t})N^{\prime }(d_1\left( k_t,\Theta_r(\hat{k}_t)\right))(T-t) \right)^{2} (\Theta_r(\hat{k}_t))^3}\nonumber\\
&=&
\frac{(\Theta_r(\hat{k}_t))^4-(I(t,T,X_t,\hat{k}_t))^4}{4\left( \exp(X_{t})N^{\prime }(d_1\left( k_t,\Theta_r(\hat{k}_t)\right)) \right)^{2} (\Theta_r(\hat{k}_t))^3},\nonumber
\end{eqnarray}
and $U_r=O((T-r)^H)$ it follows directly that $T_1=O((T-t)^{2H+1})$.

Now, let us study $T_2$. 
Towards this end, we apply the anticipating It\^{o}'s formula (\ref{Ito2}) to the process
\[
H(s,T,X_{s},\hat{k}_t,v_{s})J_{s},
\]%
where $J_{s}=\int_{s}^{T}( BS^{-1}) ^{\prime }( \hat{k}_t,\Gamma _{u})\zeta_{u}du$. Then,
taking conditional expectations we get
\begin{eqnarray*}
0 &=&E_{t}\Bigg[ H(t,T,X_t,\hat{k}_t,v_{t})J_{t}  \\
&&+\int^T_t H(s,T,X_{s},\hat{k}_t,v_{s}) dJ_{s}\\
&&+\int^T_t\frac{\partial^2}{\partial x \partial \sigma} H(s,T,X_{s},\hat{k}_t,v_{s}) J_{s} \frac{\partial v}{\partial y} (D^W_s Y_s) \sigma_s ds
\\
&&+\int^T_t \frac{\partial}{\partial x} H(s,T,X_{s},\hat{k}_t,v_{s}) (D^W_s J_s) \sigma_s ds \\
&&+\int^T_t \frac{\partial}{\partial t} H(s,T,X_{s},\hat{k}_t,v_{s})J_{s} ds \\
&&+\int^T_t \frac{\partial}{\partial \sigma} H(s,T,X_{s},\hat{k}_t,v_{s})\frac{\partial v}{\partial t} J_{s}ds  \\
&&+\int^T_t \frac{\partial}{\partial \sigma} H(s,T,X_{s},\hat{k}_t,v_{s})\frac{\partial v}{\partial y} J_{s}dY_s \\
&&+\int^T_t \frac{\partial}{\partial x} H(s,T,X_{s},\hat{k}_t,v_{s}) J_{s} dX_s \\
&&+\frac{1}{2}\int^T_t \frac{\partial^2}{\partial x^2} H(s,T,X_{s},\hat{k}_t,v_{s}) J_{s} d\langle X\rangle_s \Bigg]. 
\end{eqnarray*}
Now, using  the relationships
\begin{eqnarray*}
&&\frac{1}{\sigma(T-t)}\frac{\partial}{\partial \sigma}BS(t,T,x,k,\sigma)=\left(\frac{\partial^2}{\partial x^2} - \frac{\partial}{\partial x}\right)BS(t,T,x,k,\sigma),\\
&&\left(\frac{\partial}{\partial t} + \frac{1}{2}\sigma^2\frac{\partial^2}{\partial x^2}  - \frac{1}{2}\sigma^2 \frac{\partial}{\partial x} \right)BS(t,T,x,k,\sigma)=0,\\
&&D^W_s J_s = \rho \int_{s}^{T}( BS^{-1}) ^{\prime }( \hat{k}_t,\Gamma _{r})D^W_s \zeta_{r}dr,\\
&&D^W_s Y_s = \rho \int^T_s D^W_s \sigma^2_r dr,
\end{eqnarray*}
we obtain
\begin{eqnarray*}
0&=&E_{t}\Bigg[ H(t,T,X_t,\hat{k}_t,v_{t})J_{t}  \\
&&-\int_{t}^{T}H(s,T,X_{s},\hat{k}_t,v_{s})( BS^{-1}) ^{\prime
}( X_t,\Gamma _{s}) \zeta _{s}ds \\
&& 
+ \frac{\rho}{2} \int_{t}^{T} \left(\frac{\partial^3}{\partial x^3} - \frac{\partial^2}{\partial x^2} \right) H(s,T,X_{s},\hat{k}_t,v_{s})J_{s}\zeta_s ds
\\
&&+ \rho \int_{t}^{T}\frac{\partial }{\partial x}H(s,T,X_{s},\hat{k}_t,v_{s})\left( \int_{s}^{T}( BS^{-1}) ^{\prime }( k^{\ast}_t,\Gamma _{r}) ( D_{s}^{W}\zeta _{r}) dr\right) \sigma _{s}ds
\Bigg],
\end{eqnarray*}
which implies that
\begin{eqnarray*}
T_{2} &=&E_{t}\Bigg[ \frac{\rho }{2} H(t,T,X_t,\hat{k}_t,v_{t})J_{t}\\
&&
+ \frac{\rho^2}{4} \int_{t}^{T} \left(\frac{\partial^3}{\partial x^3} - \frac{\partial^2}{\partial x^2} \right) H(s,T,X_{s},\hat{k}_t,v_{s})J_{s}\zeta_s ds
\\
&&
+ \frac{\rho^2}{2} \int_{t}^{T}\frac{\partial }{\partial x}H(s,T,X_{s},\hat{k}_t,v_{s})\left( \int_{s}^{T}( BS^{-1}) ^{\prime }(\hat{k}_t,\Gamma _{r}) ( D_{s}^{W}\zeta _{r}) dr\right) \sigma _{s}ds\Bigg]\\
&=&T_{2}^{1}+T_{2}^{2}+T_{2}^{3}.
\end{eqnarray*}
Now, the study of $T_2$ is decomposed into two steps.

\textit{Step 1 } Notice that
\begin{eqnarray*}
H(t,T,X_{t},\hat{k}_t,v_t)
&=&\frac{e^{X_{t}}N^{\prime }(d_1(
\hat{k}_t,v_t) )}{v_t\sqrt{T-t}}\left( 1-\frac{d_1( \hat{k}_t,v_t) 
}{v_t\sqrt{T-t}}\right)\\
&=&\frac{e^{X_{t}}N^{\prime }(d_1(\hat{k}_t,v_t) )}{2v_t^3\sqrt{T-t}}\left( I_t(t,T,X_t,\hat{k}_t) - v_t^2 \right).
\end{eqnarray*}
Then

\begin{eqnarray}
\label{T21}
\lefteqn{\lim_{T\rightarrow t}T_{2}^{1}}\nonumber\\
&=&
\lim_{T\to t}\frac{\rho}{2} E_{t}\Bigg[ \frac{e^{X_{t}}N^{\prime }(d_1(
\hat{k}_t,v_t) )}{2v_t^3\sqrt{T-t}}\left( (I_t(t,T,X_t,\hat{k}_t))^2-v_t^2\right)\nonumber\\
&&\times\int_{t}^{T}\frac{1}{e^{X_{t}}N^{\prime }(d_{+}\left( \hat{k}_t,BS^{-1}(\hat{k}_t,\Gamma _{s})\right)) \sqrt{T-t}}\zeta _{s}ds \Bigg].
\end{eqnarray}%
and the norm of this is of the order $O(\nu (T-t)^{H+\frac12})$. Then, as
\begin{eqnarray*}
(I_t(t,T,X_t,\hat{k}_t))^2-v_t^2&=&(I(t,T,X_t,\hat{k}_t) +v_t)(I(t,T,X_t,\hat{k}_t) -v_t)\\
&=&(I(t,T,X_t,\hat{k}_t) +v_t)\left((I(t,T,X_t,\hat{k}_t) - E_t[v_t])+(E_t[v_t]-v_t)\right),
\end{eqnarray*}
we get
\begin{eqnarray}
\label{T21-2}
\lefteqn{\lim_{T\rightarrow t}T_{2}^{1}}\nonumber\\
&=&\lim_{T\rightarrow t}\frac{\rho}{4\sigma_t^2(T-t)}\nonumber\\
&&\times
E_t \left[ (I_t(t,T,X_t,\hat{k}_t)+v_t)\left((I(t,T,X_t,\hat{k}_t) - E_t[v_t])+(E_t[v_t]-v_t)\right)\int_{t}^{T} \int_{s}^{T}D_{s}^{W}\sigma _{r}^{2}dr ds\right]\nonumber\\
&=&\lim_{T\rightarrow t}\frac{\rho}{4\sigma_t^2(T-t)}
E_t \left[ (I_t(t,T,X_t,\hat{k}_t)+v_t)(I(\hat{k}_t) - E_t[v_t])\int_{t}^{T} \int_{s}^{T}D_{s}^{W}\sigma _{r}^{2}dr ds\right] \nonumber\\
&&+\lim_{T\rightarrow t}\frac{\rho}{4\sigma_t^2 (T-t)}
E_t \left[(I_t(t,T,X_t,\hat{k}_t)+v_t)(E_t[v_t]-v_t)\int_{t}^{T} \int_{s}^{T}D_{s}^{W}\sigma _{r}^{2}dr ds\right]\nonumber\\
&=:& \lim_{T\to t}T_2^{1,1}+\lim_{T\to t}T_2^{1,2}.
\end{eqnarray}
Notice that
\begin{eqnarray}
T_2^{1,1}&=&\lim_{T\rightarrow t}(I(t,T,X_t,\hat{k}_t) - E_t[v_t])\frac{\rho}{4\sigma_t^2(T-t)}
E_t \left[ (I_t(\hat{k}_t)+v_t)\int_{t}^{T} \int_{s}^{T}D_{s}^{W}\sigma _{r}^{2}dr ds\right]\nonumber\\
&=&(I(t,T,X_t,\hat{k}_t) - E_t[v_t])\times O((T-t)^{\frac{1}{2}+H}).
\end{eqnarray}
On the other hand,
\begin{eqnarray}
T_2^{1,2}&\leq&\lim_{T\rightarrow t}\frac{\rho}{4\sigma_t^2 (T-t)}
\left(E_t \left[\left((I_t(t,T,X_t,\hat{k}_t)+v_t)\int_{t}^{T} \int_{s}^{T}D_{s}^{W}\sigma _{r}^{2}dr ds\right)^2\right]\right)^{1/2}\nonumber\\
&&\times
\left(E_t \left[ \left(E_t[v_t]-v_t\right)^2 \right]\right)^{1/2}\nonumber\\
\end{eqnarray} Then, as 
\begin{eqnarray*}
E_t[v_t]-v_t&=&-\int_t^T E_r[D_r^W v_t]dW_r\nonumber\\
&=&-\frac{1}{2\sqrt{T-t}}\int_t^T E_r\left[\frac{\int_r^T D_r^W \sigma_s^2 ds}{v_t}\right]dW_r,
\end{eqnarray*}
and then, since $I(t,T,X_t,\hat{k}_t)+v_t<2b$ (see (H1)),
\begin{eqnarray}
\label{T212}
\lefteqn{\lim_{T\rightarrow t}T_{2}^{1,2}}\nonumber\\
&\le&
\lim_{T\rightarrow t}\frac{b\rho}{\sigma_t^2 (T-t)}
\left(E_t \left[\left(\int_{t}^{T} \int_{s}^{T}D_{s}^{W}\sigma _{r}^{2}dr ds\right)^2\right]\right)^{1/2}
\left(E_t \left[ \left(\frac{1}{2\sqrt{T-t}}\int_t^T E_r\left[\frac{\int_r^T D_r^W \sigma_s^2 ds}{v_t}\right]dW_r\right)^2 \right]\right)^{1/2}\nonumber\\
&=&
\lim_{T\rightarrow t}\frac{b\rho}{4 \sigma_t^3 (T-t)^2}
\left(E_t \left[\left(\int_{t}^{T} \int_{s}^{T}D_{s}^{W}\sigma _{r}^{2}dr ds\right)^2\right]\right)^{1/2}
\left(E_t \left[ \int_t^T \left( E_r\left[\int_r^T D_r^W \sigma_s^2 ds\right]\right)^2 dr \right]\right)^{1/2}
\nonumber\\
&=&O(T-t)^{\frac12 +2H}.
\end{eqnarray}

\textit{Step 2}. In order to see that $T_2^2$ and $T_2^3$ are $O(T-t)^{2H}$ we apply again the anticipating It\^{o}'s formula to the processes%
\begin{eqnarray*}
\left( \frac{\partial ^{3}}{\partial x^{3}}-%
\frac{\partial^2 }{\partial x^2}\right) H(s,T,X_{s},\hat{k}_t,v_{s})Z_{s},
\end{eqnarray*}
and%
\[
\frac{\partial H}{\partial x}(s,T,X_{s},\hat{k}_t,v_{s})R_{s},
\]%
where 
\begin{eqnarray*}
Z_{s}&:=&\int_{s}^{T} \zeta_u J_{u} du,\\
R_{s}&:=&\int_{s}^{T} \left(\int_{u}^{T}(BS^{-1}) ^{\prime }(  \hat{k}_t,\Gamma _{r}) (D_{s}^{W}\zeta_{r}) dr\right) \sigma _{u}du.
\end{eqnarray*}
Then we get%
\begin{eqnarray}
T_{2}^{2}
&=&\frac{\rho^2}{4}E_{t}\Bigg[ \left( \frac{\partial ^{3}}{\partial x^{3}}-\frac{\partial^2 }{\partial x^2}\right) 
H(t,T,X_{t}, \hat{k}_t,v_{t})Z_t
\nonumber \\
&&+ \frac{\rho}{2}\int_{t}^{T}\left( \frac{\partial^{3}}{\partial x^{3}}-
\frac{\partial^2 }{\partial x^2}\right)^2 H(s,T,X_{s},\hat{k}_t,v_{s}) Z_{s}\zeta_s ds  \nonumber \\
&&+\rho \int_{t}^{T}\frac{\partial}{\partial x}\left( \frac{\partial ^{3}}{\partial x^{3}}-
\frac{\partial^2 }{\partial x^2}\right) 
H(s,T,X_{s},\hat{k}_t,v_{s})(D_{s}^{W} Z_s ) \sigma_s ds\Bigg], \label{T22}
\end{eqnarray}%
and 
\begin{eqnarray}
T_{2}^{3} &=&
\frac{\rho^2}{2}E_{t}\Bigg[\frac{\partial H}{\partial x}%
(t,T,X_t,\hat{k}_t,v_{t})R_t \nonumber \\
&&+\frac{\rho }{2}\int_{t}^{T}\left( \frac{\partial ^{3}}{\partial x^{3}}-%
\frac{\partial^2 }{\partial x^2}\right) \frac{\partial H}{\partial x}%
(s,T,X_{s},\hat{k}_t,v_{s}) R_{s} \zeta_{s}ds  \nonumber \\
&&+\rho \int_{t}^{T}\frac{\partial ^{2}H}{\partial x^{2}}(s,T,X_{s},\hat{k}_t,v_{s})\nonumber\\
&&\hspace{0.5cm}\times\left(\int_{s}^{T}\int_{r}^{T}\left( 
BS^{-1}\right) ^{\prime }(  \hat{k}_t,\Gamma_{u}) (
D_{s}^{W}D_{r}^{W}\zeta_{u}) dudr\right) \sigma_{s} ds\Bigg].  \label{T23}
\end{eqnarray}%
Lemma 4.1 in Al\`{o}s, Le\'{o}n and Vives (2007) gives us that the last two terms
in (\ref{T22}) and (\ref{T23}) are $O(\nu^3(T-t)^{3H})$. 
Now, as

\begin{eqnarray*}
\lefteqn{\left|\left( \frac{\partial ^{3}}{\partial x^{3}}-\frac{\partial^2 }{\partial x^2}\right) H(t,T,X_t,\hat{k}_t,v_{t})\right|}\\
&=& 
\Bigg|- \frac{d_1\left( \hat{k}_t,v_t\right)}{v_t \sqrt{T-t}} 
\frac{e^{X_{t}}N^{\prime }(d_1\left(\hat{k}_t,v_t\right) )}{v_t\sqrt{T-t}}
\left( 1-\frac{d_1\left( \hat{k}_t,v_t\right) }{v_t\sqrt{T-t}}\right)^3
\nonumber\\
&&- \left(3 - \frac{d_1\left( \hat{k}_t,v_t\right)}{v_t\sqrt{T-t}} \right)
\frac{e^{X_{t}}N^{\prime }(d_1\left(\hat{k}_t,v_t\right) )}{(v_t\sqrt{T-t})^3}
\left( 1-\frac{d_1\left( \hat{k}_t,v_t\right) }{v_t\sqrt{T-t}}\right)
+ 3\frac{e^{X_{t}}N^{\prime }(d_1\left(\hat{k}_t,v_t\right) )}{(v_t\sqrt{T-t})^5}\Bigg|\\
&=& 
\frac{3e^{X_{t}}N^{\prime }(d_1\left(\hat{k}_t,v_t\right) )}{v_t^5}(T-t)^{-\frac{5}{2}}
+O((T-t)^{-\frac{3}{2}},
\end{eqnarray*}
and
\begin{eqnarray*}
\lefteqn{\left|\frac{\partial H}{\partial x}(t,T,X_t,\hat{k}_t,v_{t})\right|}\\
&=&\left|\frac{e^{X_{t}}N^{\prime }(d_{1}\left(
\hat{k}_t,v_t\right) )}{v_t\sqrt{T-t}}\left( 1-\frac{d_{1}\left( \hat{k}_t,v_t\right) 
}{v_t\sqrt{T-t}}\right)^2 
- \frac{e^{X_{t}}N^{\prime }(d_{1}\left(\hat{k}_t,v_t\right) )}{(v_t\sqrt{T-t})^3}\right|\\
&=&
\frac{e^{X_{t}}N^{\prime }(d_1\left(\hat{k}_t,v_t\right) )}{v_t^3}(T-t)^{-\frac{3}{2}}
+O((T-t)^{-\frac{1}{2}}.
\end{eqnarray*}
\begin{eqnarray}
\label{LT22}
\lefteqn{\lim_{T\rightarrow t}{T_{2}^{2}} }\nonumber\\
&=& \frac{\rho^2}{4}E_{t}\left[ \left( \frac{\partial ^{3}}{\partial x^{3}}-\frac{\partial^2 }{\partial x^2}\right)H(t,T,X_t,\hat{k}_t,v_{t}) Z_t\right]\nonumber\\
&=&\frac{\rho^2}{4(T-t)^3}E_{t}\Bigg[ 3\frac{e^{X_{t}}N'(d_1(\hat{k}_t,v_t)) }{v_{t}^{5}}\nonumber\\
&& \times \int_{t}^{T} \sigma_s \left(\int_{t}^{T} D_s^W \sigma _{r}^{2}dr \right) \left(\int_{s}^{T}\frac{\zeta_{r}}{e^{X_{t}}N^{\prime }(d_1( \hat{k}_t,BS^{-1}(\hat{k}_t,\Gamma _{r}))) }dr\right) ds \Bigg]\nonumber\\
&=&\lim_{T\rightarrow t}\frac{3\rho ^{2}}{4\sigma_{t}^{5}(T-t)^{3}}E_{t}\left[ \int_{t}^{T}\left( \int_{s}^{T}D_{s}^{W}\sigma _{r}^{2}dr\right) \left(\int_{s}^{T}\zeta _{r}dr\right) \sigma _{s}ds\right]\nonumber\\
&=&\lim_{T\rightarrow t}\frac{3\rho ^{2}}{4\sigma _{t}^{5}(T-t)^{3}}
E_{t}\left[\int_{t}^{T}\left( \int_{s}^{T}D_{s}^{W}\sigma _{r}^{2}dr\right) \left(\int_{s}^{T}\sigma _{r} \int_{r}^{T}D_{r}^{W}\sigma _{\theta}^{2}d\theta  dr\right) \sigma _{s}ds \right] \nonumber\\
&=&\lim_{T\rightarrow t}\frac{3\rho ^{2}}{4\sigma _{t}^{3}(T-t)^{3}}E_{t}\left[\int_{t}^{T}\left(\int_{s}^{T}D_{s}^{W}\sigma _{r}^{2}dr\right) \left( \int_{s}^{T}\int_{r}^{T}D_{r}^{W}\sigma _{\theta }^{2}d\theta dr\right) ds\right]\nonumber\\
&=&\lim_{T\rightarrow t}\frac{3\rho ^{2}}{8\sigma _{t}^{3}(T-t)^{3}}E_{t}\left[\left( \int_{t}^{T}\int_{s}^{T}D_{s}^{W}\sigma _{r}^{2} drds\right) ^{2}\right], 
\end{eqnarray}%
and 

\begin{eqnarray}
\label{TL23}
\lefteqn{\lim_{T\rightarrow t}{T_{2}^{3}}}\nonumber\\
&=&\lim_{T\rightarrow t}\frac{\rho^2}{{2}}E_{t}\left[\frac{\partial H}{\partial x}(t,T,X_t,\hat{k}_t,v_{t})R_t\right]\nonumber\\
&=&
\lim_{T\rightarrow t}\frac{\rho^2}{{2}}E_{t}\Bigg[
\frac{1}{4}\frac{e^{X_{t}}N^{\prime }(d_1( \hat{k}_t,v_t) )}{( v_t \sqrt{T-t})^{3}}\left( v_t^{2}(T-t)-4\right)\nonumber\\
&&
\times\int_{t}^{T} \int_{s}^{T}\frac{1}{e^{X_{t}}N^{\prime }(d_1(  \hat{k}_t,BS^{-1}(\hat{k}_t,\Gamma _{r}))) \sqrt{T-t}} \left(D_{s}^{W}\left(\sigma_r \int_{r}^{T} D_r^{W}\sigma _{u}^{2}du\right)\right) dr \sigma _{s}ds \Bigg]
\nonumber\\
&=&-\lim_{T\rightarrow t}\frac{\rho^{2}}{2\sigma _{t}^{2}(T-t)^{2}}
E_{t}\Bigg[ \int_{t}^{T} \int_{s}^{T}D_{s}^{W}\sigma_{r} \int_{r}^{T}D_{r}^{W}\sigma _{u}^{2}du dr ds \nonumber\\
&& +\int_{t}^{T} \int_{s}^{T} \sigma_{r}\int_{r}^{T}D_{s}^{W}D_{r}^{W}\sigma _{u}^{2}du dr ds \Bigg]\nonumber \\
&=&
-\lim_{T\rightarrow t}\frac{\rho ^{2}}{{2\sigma _{t}^2(T-t)^2}}E_{t}\left[ 
\int_{t}^{T} \int_{s}^{T}D_{s}^{W}\sigma_{r} \int_{r}^{T}D_{r}^{W}\sigma _{u}^{2}du dr ds
\right] 
\nonumber\\
&&-\lim_{T\rightarrow t}\frac{\rho^{2}}{2\sigma_{t}(T-t)^{2}}E_{t}\left[ \int_{t}^{T} \int_{s}^{T} \int_{r}^{T}D_{s}^{W}D_{r}^{W}\sigma _{u}^{2}du dr ds\right].
\end{eqnarray}
Let us now summarize the previous computations. We have seen that
\begin{eqnarray}
I(t,T,X_t,\hat{k}_t) -E_t[v_t]&=&T_1+T_2\nonumber\\
&=&T_1
+T_2^{1,1}+T_2^{1,2}
+T_2^2+T_2^3
\end{eqnarray}
where $$\lim_{T\to t}(T_1+T_2^{1,2})
=o(T-t)^{2H},$$
\begin{eqnarray*}
\lim_{T\to t}
T_2^{1,1}
&=&
\lim_{T\to t}
(I(t,T,X_t,\hat{k}_t) - E_t[v_t])\frac{\rho}{4\sigma_t^2(T-t)}
E_t \left[ (I(t,T,X_t,\hat{k}_t)+v_t)\int_{t}^{T} \int_{s}^{T}D_{s}^{W}\sigma _{r}^{2}dr ds\right],\end{eqnarray*}
\begin{eqnarray*}
\lim_{T\to t}
T_2^2
&=&
\lim_{T\to t}
\frac{3\rho ^{2}}{8\sigma _{t}^{3}(T-t)^{3}}E_{t}\left[\left( \int_{t}^{T}\int_{s}^{T}D_{s}^{W}\sigma _{r}^{2} drds\right) ^{2}\right]
+o(T-t)^{2H},
\end{eqnarray*}
and
\begin{eqnarray}
\lim_{T\to t}
T_2^3&=&
-\lim_{T\rightarrow t}\frac{\rho ^{2}}{{2\sigma _{t}^2(T-t)^2}}E_{t}\left[ 
\int_{t}^{T} \int_{s}^{T}D_{s}^{W}\sigma_{r} \int_{r}^{T}D_{r}^{W}\sigma _{u}^{2}du dr ds
\right] \nonumber\\
&&-\lim_{T\rightarrow t}\frac{\rho^{2}}{2\sigma_{t}(T-t)^{2 }}E_{t}\left[ \int_{t}^{T} \int_{s}^{T} \int_{r}^{T}D_{s}^{W}D_{r}^{W}\sigma _{u}^{2}du dr ds\right]\nonumber\\
&&+o(T-t)^{2H}.
\end{eqnarray}
Then, as there is some $\epsilon$ such that, if $T-t<\epsilon$ 
$$
\left|\frac{\rho}{4\sigma_t^2(T-t)}
E_t \left( (I_t(t,T,X_t,\hat{k}_t)+v_t)\int_{t}^{T} \int_{s}^{T}D_{s}^{W}\sigma _{r}^{2}dr ds\right)\right|<1
$$
we can write
\begin{eqnarray}
\lim_{T\rightarrow t} \frac{I(t,T,X_t,\hat{k}_t) -E_t[v_t]}{(T-t)^{2H}}&=&\lim_{T\rightarrow t}\frac{1}{(T-t)^{2H}}\frac{T_1+T_2^{1,2}+T_2^2+T_3^3}{1-\frac{\rho}{4\sigma_t^2(T-t)}
E_t \left( (I_t(t,T,X_t,\hat{k}_t)+v_t)\int_{t}^{T} \int_{s}^{T}D_{s}^{W}\sigma _{r}^{2}dr ds\right)}\nonumber\\
&=&\lim_{T\rightarrow t}\frac{3\rho ^{2}}{8\sigma _{t}^{3}(T-t)^{3+2H}}E_{t}\left[\left( \int_{t}^{T}\int_{s}^{T}D_{s}^{W}\sigma _{r}^{2} drds\right) ^{2}\right]\nonumber\\
&&
-\lim_{T\rightarrow t}\frac{\rho ^{2}}{{2\sigma _{t}^2(T-t)^{2+2H}}}E_{t}\left[ 
\int_{t}^{T} \int_{s}^{T}D_{s}^{W}\sigma_{r} \int_{r}^{T}D_{r}^{W}\sigma _{u}^{2}du dr ds
\right] \nonumber\\
&&-\lim_{T\rightarrow t}\frac{\rho^{2}}{{2\sigma_{t}(T-t)^{2+2H }}}E_{t}\left[ \int_{t}^{T} \int_{s}^{T} \int_{r}^{T}D_{s}^{W}D_{r}^{W}\sigma _{u}^{2}du dr ds\right],
\end{eqnarray}
as we wanted to prove.
\end{proof}

\section{Standard deviations of Monte Carlo simulations}\label{appendix3}

We list the standard deviations of the Monte Carlo method used in Section \ref{section4}.
The standard deviation of the implied volatility cannot be obtained directly.
Therefore, we approximate the standard deviation of the implied volatility by that of the option premium divided by the Vega value.
Since the standard deviation of ``AS(4.8)'' is regarded equal to that of ``ATMI'' in this computation, we summarize them in the correlated cases.

\begin{table}[htpb]
\caption{Standard deviations of the Monte Carlo simulations ($\rho=0$, $\sigma_0=20\%$, $\alpha=0.8$)}
\label{table:t7}
\newcolumntype{Y}{>{\centering\arraybackslash}X}
\newcolumntype{Z}{>{\raggedleft\arraybackslash}X}
\begin{tabularx}{\linewidth}{XXYYY} \hline
$H$  &  $T$    & VS   & IV($\hat{k}$)      & ATMI    \\ \hline\hline
0.1 & 0.25 & 0.001\% & 0.002\% & 0.002\%  \\
    & 0.5  & 0.001\% & 0.002\% & 0.002\%  \\
    & 1.0  & 0.001\% & 0.002\% & 0.002\%  \\
    & 2.0  & 0.001\% & 0.002\% & 0.002\%  \\
~   & 3.0  & 0.001\% & 0.002\% & 0.003\%  \\ \hdashline
0.3 & 0.25 & 0.001\% & 0.002\% & 0.002\%  \\
    & 0.5  & 0.001\% & 0.002\% & 0.002\%  \\
    & 1.0  & 0.001\% & 0.002\% & 0.002\%  \\
    & 2.0  & 0.001\% & 0.002\% & 0.002\%  \\
~   & 3.0  & 0.001\% & 0.003\% & 0.003\%  \\ \hdashline
0.5 & 0.25 & 0.001\% & 0.001\% & 0.001\%  \\
    & 0.5  & 0.001\% & 0.001\% & 0.001\%  \\
    & 1.0  & 0.001\% & 0.002\% & 0.002\%  \\
    & 2.0  & 0.001\% & 0.002\% & 0.003\%  \\
~   & 3.0  & 0.002\% & 0.003\% & 0.003\%  \\ \hdashline
0.7 & 0.25 & 0.000\% & 0.001\% & 0.001\%  \\
    & 0.5  & 0.001\% & 0.001\% & 0.001\%  \\
    & 1.0  & 0.001\% & 0.002\% & 0.002\%  \\
    & 2.0  & 0.001\% & 0.003\% & 0.003\%  \\
~   & 3.0  & 0.002\% & 0.003\% & 0.003\%  \\ \hdashline
0.9 & 0.25 & 0.000\% & 0.001\% & 0.001\%  \\
    & 0.5  & 0.000\% & 0.001\% & 0.001\%  \\
    & 1.0  & 0.001\% & 0.002\% & 0.002\%  \\
    & 2.0  & 0.002\% & 0.003\% & 0.003\%  \\
~   & 3.0  & 0.002\% & 0.003\% & 0.003\% \\ \hline
\end{tabularx}
\end{table}

\begin{table}[htpb]
\caption{Standard deviations of the Monte Carlo simulations ($\rho=-0.8$, $\sigma_0=20\%$, $\alpha=0.8$)}
\label{table:t8}
\newcolumntype{Y}{>{\centering\arraybackslash}X}
\newcolumntype{Z}{>{\raggedleft\arraybackslash}X}
\begin{tabularx}{\linewidth}{XXYYY} \hline
$H$   & $T$    & VS   &  IV($\hat{k}$)     & ATMI \& AS(4.8) \\ \hline\hline
0.1 & 0.25 & 0.001\% & 0.002\% & 0.002\% \\
    & 0.5  & 0.001\% & 0.002\% & 0.002\% \\
    & 1.0  & 0.001\% & 0.002\% & 0.002\% \\
    & 2.0  & 0.001\% & 0.002\% & 0.002\% \\
~   & 3.0  & 0.001\% & 0.002\% & 0.002\% \\\hdashline
0.3 & 0.25 & 0.001\% & 0.001\% & 0.001\% \\
    & 0.5  & 0.001\% & 0.002\% & 0.002\% \\
    & 1.0  & 0.001\% & 0.002\% & 0.002\% \\
    & 2.0  & 0.001\% & 0.002\% & 0.002\% \\
~   & 3.0  & 0.001\% & 0.002\% & 0.002\% \\\hdashline
0.5 & 0.25 & 0.001\% & 0.001\% & 0.001\% \\
    & 0.5  & 0.001\% & 0.001\% & 0.001\% \\
    & 1.0  & 0.001\% & 0.002\% & 0.002\% \\
    & 2.0  & 0.001\% & 0.002\% & 0.002\% \\
~   & 3.0  & 0.002\% & 0.002\% & 0.002\% \\\hdashline
0.7 & 0.25 & 0.000\% & 0.001\% & 0.001\% \\
    & 0.5  & 0.001\% & 0.001\% & 0.001\% \\
    & 1.0  & 0.001\% & 0.002\% & 0.002\% \\
    & 2.0  & 0.001\% & 0.002\% & 0.002\% \\
~   & 3.0  & 0.002\% & 0.003\% & 0.003\% \\\hdashline
0.9 & 0.25 & 0.000\% & 0.000\% & 0.000\% \\
    & 0.5  & 0.000\% & 0.001\% & 0.001\% \\
    & 1.0  & 0.001\% & 0.001\% & 0.001\% \\
    & 2.0  & 0.002\% & 0.002\% & 0.002\% \\
~   & 3.0  & 0.002\% & 0.003\% & 0.003\% \\ \hline
\end{tabularx}
\end{table}

\begin{table}[htpb]
\caption{Standard deviations of the Monte Carlo simulations ($\rho=0$, $\sigma_0=40\%$, $\alpha=0.8$)}
\label{table:t9}
\newcolumntype{Y}{>{\centering\arraybackslash}X}
\newcolumntype{Z}{>{\raggedleft\arraybackslash}X}
\begin{tabularx}{\linewidth}{XXYYY} \hline
$H$  &  $T$    & VS   & IV($\hat{k}$)      & ATMI    \\ \hline\hline
0.1 & 0.25 & 0.002\% & 0.004\% & 0.004\%  \\
    & 0.5  & 0.002\% & 0.004\% & 0.004\%  \\
    & 1.0  & 0.002\% & 0.004\% & 0.004\%  \\
    & 2.0  & 0.002\% & 0.004\% & 0.004\%  \\
~   & 3.0  & 0.002\% & 0.004\% & 0.004\%  \\\hdashline
0.3 & 0.25 & 0.001\% & 0.003\% & 0.003\%  \\
    & 0.5  & 0.002\% & 0.003\% & 0.003\%  \\
    & 1.0  & 0.002\% & 0.004\% & 0.004\%  \\
    & 2.0  & 0.003\% & 0.004\% & 0.004\%  \\
~   & 3.0  & 0.003\% & 0.005\% & 0.005\%  \\\hdashline
0.5 & 0.25 & 0.001\% & 0.002\% & 0.002\%  \\
    & 0.5  & 0.001\% & 0.003\% & 0.003\%  \\
    & 1.0  & 0.002\% & 0.003\% & 0.003\%  \\
    & 2.0  & 0.003\% & 0.004\% & 0.004\%  \\
~   & 3.0  & 0.003\% & 0.005\% & 0.005\%  \\\hdashline
0.7 & 0.25 & 0.001\% & 0.001\% & 0.001\%  \\
    & 0.5  & 0.001\% & 0.002\% & 0.002\%  \\
    & 1.0  & 0.002\% & 0.003\% & 0.003\%  \\
    & 2.0  & 0.003\% & 0.005\% & 0.005\%  \\
~   & 3.0  & 0.004\% & 0.005\% & 0.005\%  \\\hdashline
0.9 & 0.25 & 0.001\% & 0.001\% & 0.001\%  \\
    & 0.5  & 0.001\% & 0.002\% & 0.002\%  \\
    & 1.0  & 0.002\% & 0.003\% & 0.003\%  \\
    & 2.0  & 0.003\% & 0.005\% & 0.005\%  \\
~   & 3.0  & 0.004\% & 0.006\% & 0.006\% \\\hline
\end{tabularx}
\end{table}

\begin{table}[htpb]
\caption{Standard deviations of the Monte Carlo simulations ($\rho=-0.8$, $\sigma_0=40\%$, $\alpha=0.8$)}
\label{table:t10}
\newcolumntype{Y}{>{\centering\arraybackslash}X}
\newcolumntype{Z}{>{\raggedleft\arraybackslash}X}
\begin{tabularx}{\linewidth}{XXYYY} \hline
$H$   & $T$    & VS   &  IV($\hat{k}$)     & ATMI \& AS(4.8) \\ \hline\hline
0.1 & 0.25 & 0.002\% & 0.004\% & 0.004\% \\
    & 0.5  & 0.002\% & 0.004\% & 0.004\% \\
    & 1.0  & 0.002\% & 0.004\% & 0.004\% \\
    & 2.0  & 0.002\% & 0.004\% & 0.004\% \\
~   & 3.0  & 0.002\% & 0.003\% & 0.004\% \\\hdashline
0.3 & 0.25 & 0.001\% & 0.003\% & 0.003\% \\
    & 0.5  & 0.002\% & 0.003\% & 0.003\% \\
    & 1.0  & 0.002\% & 0.003\% & 0.003\% \\
    & 2.0  & 0.003\% & 0.003\% & 0.004\% \\
~   & 3.0  & 0.003\% & 0.003\% & 0.004\% \\\hdashline
0.5 & 0.25 & 0.001\% & 0.002\% & 0.002\% \\
    & 0.5  & 0.001\% & 0.002\% & 0.002\% \\
    & 1.0  & 0.002\% & 0.003\% & 0.003\% \\
    & 2.0  & 0.003\% & 0.003\% & 0.004\% \\
~   & 3.0  & 0.003\% & 0.004\% & 0.004\% \\\hdashline
0.7 & 0.25 & 0.001\% & 0.001\% & 0.001\% \\
    & 0.5  & 0.001\% & 0.002\% & 0.002\% \\
    & 1.0  & 0.002\% & 0.003\% & 0.003\% \\
    & 2.0  & 0.003\% & 0.004\% & 0.004\% \\
~   & 3.0  & 0.004\% & 0.004\% & 0.004\% \\\hdashline
0.9 & 0.25 & 0.001\% & 0.001\% & 0.001\% \\
    & 0.5  & 0.001\% & 0.001\% & 0.002\% \\
    & 1.0  & 0.002\% & 0.002\% & 0.002\% \\
    & 2.0  & 0.003\% & 0.004\% & 0.004\% \\
~   & 3.0  & 0.004\% & 0.005\% & 0.005\% \\\hline
\end{tabularx}
\end{table}

\begin{table}[htpb]
\caption{Standard deviations of the Monte Carlo simulations ($\rho=0$, $\sigma_0=20\%$, $\alpha=2$)}
\label{table:t11}
\newcolumntype{Y}{>{\centering\arraybackslash}X}
\newcolumntype{Z}{>{\raggedleft\arraybackslash}X}
\begin{tabularx}{\linewidth}{XXYYY} \hline
$H$  &  $T$    & VS   & IV($\hat{k}$)      & ATMI    \\ \hline\hline
0.1 & 0.25 & 0.002\% & 0.005\% & 0.005\%  \\
    & 0.5  & 0.002\% & 0.005\% & 0.005\%  \\
    & 1.0  & 0.002\% & 0.005\% & 0.005\%  \\
    & 2.0  & 0.002\% & 0.005\% & 0.005\%  \\
~   & 3.0  & 0.002\% & 0.005\% & 0.005\%  \\\hdashline
0.3 & 0.25 & 0.002\% & 0.004\% & 0.004\%  \\
    & 0.5  & 0.002\% & 0.004\% & 0.004\%  \\
    & 1.0  & 0.002\% & 0.005\% & 0.005\%  \\
    & 2.0  & 0.003\% & 0.005\% & 0.005\%  \\
~   & 3.0  & 0.003\% & 0.005\% & 0.005\%  \\\hdashline
0.5 & 0.25 & 0.001\% & 0.002\% & 0.003\%  \\
    & 0.5  & 0.002\% & 0.003\% & 0.003\%  \\
    & 1.0  & 0.002\% & 0.004\% & 0.004\%  \\
    & 2.0  & 0.003\% & 0.005\% & 0.005\%  \\
~   & 3.0  & 0.003\% & 0.005\% & 0.005\%  \\\hdashline
0.7 & 0.25 & 0.001\% & 0.002\% & 0.002\%  \\
    & 0.5  & 0.001\% & 0.003\% & 0.003\%  \\
    & 1.0  & 0.002\% & 0.004\% & 0.004\%  \\
    & 2.0  & 0.003\% & 0.005\% & 0.005\%  \\
~   & 3.0  & 0.003\% & 0.005\% & 0.005\%  \\\hdashline
0.9 & 0.25 & 0.001\% & 0.001\% & 0.001\%  \\
    & 0.5  & 0.001\% & 0.002\% & 0.002\%  \\
    & 1.0  & 0.002\% & 0.004\% & 0.004\%  \\
    & 2.0  & 0.003\% & 0.005\% & 0.005\%  \\
~   & 3.0  & 0.003\% & 0.005\% & 0.005\% \\\hline
\end{tabularx}
\end{table}

\begin{table}[htpb]
\caption{Standard deviations of the Monte Carlo simulations ($\rho=-0.8$, $\sigma_0=20\%$, $\alpha=2$)}
\label{table:t12}
\newcolumntype{Y}{>{\centering\arraybackslash}X}
\newcolumntype{Z}{>{\raggedleft\arraybackslash}X}
\begin{tabularx}{\linewidth}{XXYYY} \hline
$H$   & $T$    & VS   &  IV($\hat{k}$)     & ATMI \& AS(4.8) \\ \hline\hline
0.1 & 0.25 & 0.002\% & 0.005\% & 0.005\% \\
    & 0.5  & 0.002\% & 0.005\% & 0.005\% \\
    & 1.0  & 0.002\% & 0.005\% & 0.005\% \\
    & 2.0  & 0.002\% & 0.005\% & 0.005\% \\
~   & 3.0  & 0.002\% & 0.004\% & 0.005\% \\ \hdashline
0.3 & 0.25 & 0.002\% & 0.003\% & 0.003\% \\
    & 0.5  & 0.002\% & 0.004\% & 0.004\% \\
    & 1.0  & 0.002\% & 0.004\% & 0.004\% \\
    & 2.0  & 0.003\% & 0.004\% & 0.004\% \\
~   & 3.0  & 0.003\% & 0.004\% & 0.004\% \\ \hdashline
0.5 & 0.25 & 0.001\% & 0.002\% & 0.002\% \\
    & 0.5  & 0.002\% & 0.003\% & 0.003\% \\
    & 1.0  & 0.002\% & 0.004\% & 0.004\% \\
    & 2.0  & 0.003\% & 0.004\% & 0.004\% \\
~   & 3.0  & 0.003\% & 0.004\% & 0.004\% \\ \hdashline
0.7 & 0.25 & 0.001\% & 0.002\% & 0.002\% \\
    & 0.5  & 0.001\% & 0.002\% & 0.003\% \\
    & 1.0  & 0.002\% & 0.003\% & 0.003\% \\
    & 2.0  & 0.003\% & 0.004\% & 0.004\% \\
~   & 3.0  & 0.003\% & 0.004\% & 0.004\% \\ \hdashline
0.9 & 0.25 & 0.001\% & 0.001\% & 0.001\% \\
    & 0.5  & 0.001\% & 0.002\% & 0.002\% \\
    & 1.0  & 0.002\% & 0.003\% & 0.003\% \\
    & 2.0  & 0.003\% & 0.004\% & 0.004\% \\
~   & 3.0  & 0.003\% & 0.004\% & 0.005\% \\ \hline
\end{tabularx}
\end{table}

\end{document}